\documentclass[aps,amsmath,amssymb,twocolumn]{revtex4}
\usepackage{graphicx}
\usepackage{pgfplots}
\pgfplotsset{compat=1.17}
\usepackage{dcolumn}
\usepackage{color}
\usepackage{scrextend}
\usepackage{subfig}
\usepackage{bm}
\usepackage{amsmath,amssymb,graphicx}
\usepackage{epsfig}
\usepackage{enumerate}
\usepackage{array}
\usepackage{multirow}
\usepackage{tikz}
\usepackage{ragged2e}
\usepackage{textgreek}
\usepackage[normalem]{ulem}
 \usetikzlibrary{arrows.meta,positioning,fit,calc}
\usepackage{lipsum}
\usepackage{ntheorem}
\usepackage{xcolor}
\usepackage{hyperref}

\definecolor{KTBlue}{HTML}{2166AC}
\definecolor{KTBlueFill}{HTML}{EAF2F8}
\definecolor{KTOrange}{HTML}{C4511B}
\definecolor{KTOrangeFill}{HTML}{FAEEE8}
\definecolor{KTInk}{HTML}{202124}
\definecolor{KTMuted}{HTML}{687078}
\definecolor{KTRule}{HTML}{D5DAE0}
\definecolor{NHBlue}{HTML}{22577A}
\definecolor{PrepOrange}{HTML}{D97757}
\definecolor{ReadTeal}{HTML}{238B84}
\definecolor{Ink}{HTML}{24313F}
\definecolor{SoftGray}{HTML}{F4F6F8}

\newtheorem{prop}{Proposition}

\begin{document}
\title{Krylov Tomography and Finite-Uncertainty\\ Certification of Exceptional-Point Dynamics}
\author{Aritra Ghosh\footnote{aritraghosh500@gmail.com} and M. Bhattacharya}
\affiliation{School of Physics and Astronomy, Rochester Institute of Technology, 84 Lomb Memorial Drive, Rochester, New York 14623, USA}
\vskip-2.8cm
\date{\today}
\vskip-0.9cm

\vspace{5mm}
\begin{abstract}
Exceptional points (EPs) can produce striking responses, but locating one does not reveal how much of its dynamics an excitation accesses, which responses a detector distinguishes, or whether a missing signal is absent or undetected. We introduce Krylov tomography, a preparation- and measurement-aware framework that uses time-resolved data and models with stated uncertainty limits to answer these questions, while also bounding offset-induced departures from exact-EP behavior. As an explicit illustration of the general framework, we present finite-precision simulations of a red-sideband optomechanical model, whose second-moment coherence sector contains a third-order EP, certifying preparation-sensitive access to two and three directions. Krylov tomography thus bridges EP structure and finite-precision measurements, providing a general framework for non-Hermitian systems.
\end{abstract}
\maketitle

\textit{Introduction.} Non-Hermitian generators arise in diverse open systems \cite{Rotter_2009}. At an exceptional point (EP), multiple eigenvalues and eigenvectors coalesce, and the generator develops a nontrivial Jordan structure \cite{Berry_2004,Heiss_2012,Bagarello_2015,Longhi_2017,Bergholtz_2017}. Such degeneracies have been observed experimentally \cite{Dembowski_2001,Lee_2009,Kim_2016,Bergman_2021,Xiao_2021,Soleymani_2022} and are commonly probed through spectral coalescence, response poles, and polynomial-in-time transients \cite{Cartarius_2011,Minganti_2019}. In selected systems, the eigenvector structure can also be accessed \cite{Naghiloo_2019,ChenJordan_2020,ChenTomography_2025}. These probes, together with the system model, establish the exceptional structure available in the generator but not automatically the portion of the Jordan chain actually accessed by a particular experiment, since an EP can produce different pole orders in the measured response under different excitation and detection channels \cite{Hashemi_2022}, while its response strength is not determined by its order alone \cite{Wiersig_2022}. 

An experiment comprises preparation, evolution, and readout [Fig. \ref{fig1}]. Preparation determines the activated depth, defined as the number of successive nonzero vectors obtained by repeatedly applying the generator shifted by the EP eigenvalue to the prepared state \cite{Horn_2013}, forming the preparation-generated Krylov chain. Readout maps each chain vector to an output vector, and the reached directions are independently resolved when those outputs are linearly independent. A nonzero direction with vanishing measured image is detector-dark, and the resulting apparent endpoint is detector blindness. Thus the same generator can yield preparation-dependent activated depths [Fig. \ref{fig2}] and readout-dependent numbers of resolved directions. Because the generator's EP order can overstate the dynamics actually excited and resolved, without preparation- and readout-aware finite-error tests, neither the activated depth nor whether a missing response reflects chain termination or detector blindness can be established.

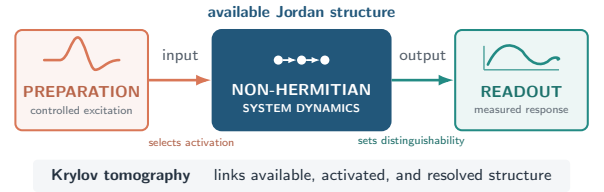
\begin{figure}[t]
\centering
\resizebox{0.90\columnwidth}{!}{
\begin{tikzpicture}[
    x=1cm,
    y=1cm,
    >=Latex,
    font=\sffamily,
    prepbox/.style={
        draw=PrepOrange,
        fill=PrepOrange!8,
        line width=0.9pt,
        rounded corners=3pt,
        minimum width=2.05cm,
        minimum height=1.55cm,
        inner sep=4pt,
        align=center
    },
    systembox/.style={
        draw=NHBlue,
        fill=NHBlue,
        line width=0.9pt,
        rounded corners=4pt,
        minimum width=2.75cm,
        minimum height=1.55cm,
        inner sep=4pt,
        align=center,
        text=white
    },
    readbox/.style={
        draw=ReadTeal,
        fill=ReadTeal!8,
        line width=0.9pt,
        rounded corners=3pt,
        minimum width=2.05cm,
        minimum height=1.55cm,
        inner sep=4pt,
        align=center
    },
    flow/.style={
        -{Latex[length=2.4mm,width=1.7mm]},
        line width=1.15pt
    },
    smalllabel/.style={
        font=\sffamily\scriptsize,
        text=Ink!78,
        align=center
    }
]

\node[prepbox]   (prep) at (0,0) {};
\node[systembox] (sys)  at (3.40,0) {};
\node[readbox]   (read) at (6.80,0) {};

\draw[
    PrepOrange,
    line width=1.15pt,
    line cap=round,
    line join=round
]
    (-0.58,0.31)
    --
    (-0.34,0.31)
    .. controls (-0.21,0.31) and (-0.16,0.66)
       .. (-0.06,0.66)
    .. controls (0.04,0.66) and (0.12,0.16)
       .. (0.21,0.16)
    .. controls (0.31,0.16) and (0.39,0.31)
       .. (0.58,0.31);

\foreach \x in {3.04,3.40,3.76}{
    \fill[white] (\x,0.31) circle (2pt);
}

\draw[
    white!88,
    line width=0.8pt,
    -{Latex[length=1.35mm,width=1mm]}
]
    (3.11,0.31)
    --
    (3.33,0.31);

\draw[
    white!88,
    line width=0.8pt,
    -{Latex[length=1.35mm,width=1mm]}
]
    (3.47,0.31)
    --
    (3.69,0.31);

\draw[ReadTeal!78,line width=0.65pt]
    (6.22,0.16) -- (6.22,0.53)
    (6.22,0.16) -- (7.38,0.16);

\draw[
    ReadTeal,
    line width=1.15pt,
    line cap=round
]
    (6.28,0.24)
    .. controls (6.48,0.66) and (6.71,0.54)
       .. (6.82,0.44)
    .. controls (7.00,0.16) and (7.18,0.26)
       .. (7.22,0.30)
    .. controls (7.31,0.39) and (7.42,0.33)
       .. (7.36,0.32);

\node[
    font=\sffamily\scriptsize\bfseries,
    text=PrepOrange!85!black
]
    at (0,-0.18)
    {PREPARATION};

\node[
    font=\sffamily\tiny,
    text=Ink!72
]
    at (0,-0.47)
    {controlled excitation};

\node[
    font=\sffamily\scriptsize\bfseries,
    text=white
]
    at (3.40,-0.14)
    {NON-HERMITIAN};

\node[
    font=\sffamily\tiny\bfseries,
    text=white!88
]
    at (3.40,-0.43)
    {SYSTEM DYNAMICS};

\node[
    font=\sffamily\scriptsize\bfseries,
    text=ReadTeal!82!black
]
    at (6.80,-0.18)
    {READOUT};

\node[
    font=\sffamily\tiny,
    text=Ink!72
]
    at (6.80,-0.47)
    {measured response};

\draw[flow,draw=PrepOrange]
    (prep.east)
    --
    (sys.west)
    node[midway,above=4pt,smalllabel]
    {input};

\draw[flow,draw=ReadTeal]
    (sys.east)
    --
    (read.west)
    node[midway,above=4pt,smalllabel]
    {output};

\node[
    font=\sffamily\tiny,
    text=PrepOrange!80!black
]
    at (1.70,-0.95)
    {selects activation};

\node[
    font=\sffamily\tiny,
    text=ReadTeal!78!black
]
    at (5.10,-0.95)
    {sets distinguishability};

\node[
    font=\sffamily\scriptsize\bfseries,
    text=NHBlue
]
    at (3.40,1.06)
    {available Jordan structure};

\node[
    fill=SoftGray,
    rounded corners=2pt,
    inner xsep=8pt,
    inner ysep=4pt,
    text=Ink
]
    at (3.40,-1.48)
    {
        \sffamily\scriptsize\bfseries
        Krylov tomography
        \quad
        \mdseries
        links available, activated, and resolved structure
    };

\end{tikzpicture}
}
\caption{\justifying{\textbf{Preparation-evolution-readout structure.} The generator fixes available Jordan chains, preparation fixes how many successive directions are reached, and readout fixes which reached directions produce linearly independent measured signals.}}
\label{fig1}
\end{figure}

To address this important operational gap, we develop Krylov tomography in finite-dimensional closed moment sectors, i.e., where the moments evolve among themselves, permitting finite-matrix predictions. Such sectors are pervasive in Gaussian bosonic and fermionic systems and in symmetry-resolved open spin dynamics \cite{Prosen_2008,Barthel_2022,Buca_2012}. Whereas Krylov methods also underpin estimation of quantum Fisher information \cite{Zhang_2025}, transport-based tomography \cite{Bourgeois_2026}, and observability-based dynamical quantum tomography \cite{Peruzzo_2026}, our central advance is, to our knowledge, the first finite-error certification theory for preparation- and readout-selected Jordan dynamics. That is, it certifies, under stated uncertainty bounds, which directions are reached and independently resolved, and whether the response terminates within tolerance.

Using time-polynomial coefficients with calibrated dynamics and readout, the method separates available Jordan order, activated depth, visible depth, and independently resolved depth, defined below. A positive detector-sensitivity lower bound rules out blindness among allowed endpoint directions. Simultaneous coefficient bounds lower-bound activated depth and certify independent resolution; a nonzero next coefficient certifies nontermination, whereas readout-gain and characteristic-rate bounds upper-bound the normalized termination residual. Because uncertainties in model parameters and approximately constant offsets  impede establishing and maintaining exact-EP operation \cite{Bergman_2021,Soleymani_2022,Mortensen_2018}, we also bound departures from it analytically. Here ``exact EP'' means exact eigenvalue-eigenvector coalescence, whereas ``near EP'' denotes a nearby point where the same modes remain trackable without assuming exact coalescence. ``Calibration'' means model information obtained independently of the tested record with stated uncertainty bounds; quantities inferred from that record instead have simultaneous uncertainty bounds. The exact- and near-EP formulations apply whenever the relevant closed sector, isolated EP contribution, and detector response are independently established.

\begin{figure}[t!]
\centering
\begin{tikzpicture}[
    >=Latex,
    font=\sffamily,
    x=1cm,y=1cm,
    state/.style={
        draw=black!65,
        rounded corners=2pt,
        minimum width=0.95cm,
        minimum height=0.50cm,
        align=center,
        font=\footnotesize,
        fill=black!3
    },
    active/.style={
        state,
        draw=blue!60!black,
        very thick,
        fill=blue!7
    },
    inactive/.style={
        state,
        draw=black!20,
        text=black!35,
        fill=black!2
    },
    zero/.style={
        draw=black!40,
        rounded corners=2pt,
        minimum width=0.50cm,
        minimum height=0.46cm,
        align=center,
        font=\footnotesize,
        fill=white
    },
    arrow/.style={->, thick, draw=black!60},
    bluearrow/.style={->, very thick, draw=blue!60!black},
    rowlab/.style={font=\scriptsize, align=right},
    note/.style={
        draw=black!18,
        rounded corners=3pt,
        fill=black!2,
        align=left,
        inner sep=4pt,
        font=\scriptsize,
        text width=6.3cm
    }
]

\node[rowlab, anchor=east] at (1.05,-4.20) {prepare\\$e_3$};
\node[active] (b13) at (1.95,-4.20) {$e_3$};
\node[active] (b12) at (3.20,-4.20) {$e_2$};
\node[active] (b11) at (4.45,-4.20) {$e_1$};
\node[zero]   (b10) at (5.45,-4.20) {$0$};
\draw[bluearrow] (b13) -- (b12);
\draw[bluearrow] (b12) -- (b11);
\draw[bluearrow] (b11) -- (b10);
\node[font=\scriptsize\bfseries, text=blue!60!black, anchor=west]
at (5.85,-4.20) {activated depth 3};

\node[rowlab, anchor=east] at (1.05,-5.20) {prepare\\$e_2$};
\node[inactive] (b23) at (1.95,-5.20) {$e_3$};
\node[active]   (b22) at (3.20,-5.20) {$e_2$};
\node[active]   (b21) at (4.45,-5.20) {$e_1$};
\node[zero]     (b20) at (5.45,-5.20) {$0$};
\draw[arrow,dashed] (b23) -- (b22);
\draw[bluearrow] (b22) -- (b21);
\draw[bluearrow] (b21) -- (b20);
\node[font=\scriptsize\bfseries, anchor=west]
at (5.85,-5.20) {activated depth 2};

\node[rowlab, anchor=east] at (1.05,-6.20) {prepare\\$e_1$};
\node[inactive] (b33) at (1.95,-6.20) {$e_3$};
\node[inactive] (b32) at (3.20,-6.20) {$e_2$};
\node[active]   (b31) at (4.45,-6.20) {$e_1$};
\node[zero]     (b30) at (5.45,-6.20) {$0$};
\draw[arrow,dashed] (b33) -- (b32);
\draw[arrow,dashed] (b32) -- (b31);
\draw[bluearrow] (b31) -- (b30);
\node[font=\scriptsize\bfseries, anchor=west]
at (5.85,-6.20) {activated depth 1};

\end{tikzpicture}

\caption{\justifying In a closed moment sector, an available size-three Jordan block with generalized eigenvectors $e_j$ can have activated depths three, two, or one. The generator fixes available order but preparation dictates which directions are reached.}
\label{fig2}
\end{figure}
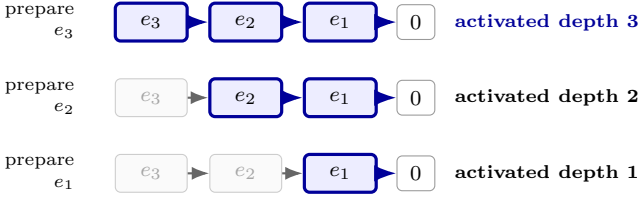

We illustrate our framework with a linearized red-sideband cavity-optomechanical model, where the resonant beam-splitter interaction coherently swaps excitation quanta, namely, cavity photons and mechanical phonons \cite{Aspelmeyer_2014}. Optomechanics provides a versatile platform for investigating EPs, both theoretically \cite{Jing_2017,Ghosh_2026} as well as experimentally \cite{Xu_2016}. A second-order EP (EP2) in the field-amplitude dynamics induces a size-three Jordan block in the number-conserving second-moment sector formed by the two populations and one cavity-mechanics coherence quadrature, while the orthogonal coherence quadrature forms a decoupled size-one block, giving a $3\oplus1$ partition, i.e., a multiblock structure \cite{Shiralieva_2026}. A cavity-noise preparation that excites only the cavity occupation activates depth three, whereas the difference of separate cavity- and mechanical-excitation runs activates depth two. The former can be realized by driving the two cavity quadratures with independent, identically distributed zero-mean Gaussian noises, yielding zero average field without preferred optical phase. Related band-limited microwave-noise injection has earlier been demonstrated experimentally \cite{Wang_2024}. Finite-imprecision simulations from a three-output model illustrate certification of three and two independently resolved directions, respectively, and yield distinct termination bounds.

\textit{Preparation access and readout resolution.}
A Jordan block affects an experiment only through the sequence of directions generated from its preparation. After the preparation pulse and subtraction of any stationary background, on a closed moment sector $\mathcal{M}$, let $x(t)\in\mathcal{M}$ denote the vector of retained excess moments and let $A:\mathcal{M}\to\mathcal{M}$ denote the time-independent drift matrix governing their closed linear evolution $\dot{x}=Ax$. For an isolated EP eigenvalue $\lambda$, let $\Pi_\lambda$ denote the spectral projector that extracts the contribution associated with $\lambda$, and let $\mathcal{G}_\lambda=\operatorname{ran}\Pi_\lambda$ be the corresponding generalized eigenspace. We define $N_\lambda=(A-\lambda I)|_{\mathcal{G}_\lambda}$ and let $\nu_\lambda=\min\{\nu\geq1:N_\lambda^\nu=0\}$ be the nilpotency index of $N_\lambda$, i.e., the size of the largest Jordan block at $\lambda$. For a prepared initial moment vector $x_0$, we let $x_\lambda=\Pi_\lambda x_0\neq0$ denote its component in $\mathcal{G}_\lambda$. The projected preparation $x_\lambda$ and a fixed readout map $C:\mathcal{M}\to\mathcal{Z}$ to a finite-dimensional output space $\mathcal{Z}$ generate the Krylov vectors $q_j=N_\lambda^jx_\lambda$ and their measured images $p_j=Cq_j$. The closure of $\mathcal{M}$, isolation of the $\lambda$ contribution, and the calibration of $C$ are independent modeling inputs. With $z_\lambda(t)=Cx_\lambda(t)$, the isolated state and output trajectories are
\begin{equation}
x_\lambda(t)
=
e^{\lambda t}
\sum_{j=0}^{\nu_\lambda-1}\frac{t^j}{j!}q_j,
\quad \quad
z_\lambda(t)
=
e^{\lambda t}
\sum_{j=0}^{\nu_\lambda-1}\frac{t^j}{j!}p_j .
\label{main_krylov_response}
\end{equation}
For sampling times $t_\ell$, $\ell=0,\ldots,L-1$, the polynomial sampling matrix is $V_{\ell j}=t_\ell^j/j!$, with $j=0,\ldots,\nu_\lambda-1$, so coefficients $p_j$ are identifiable whenever $V$ has full column rank. The preparation-activated depth is $d_{\rm act} = \min\{r\geq1:q_r=0\}$, i.e., $d_{\rm act}=r$ $\Longleftrightarrow$ $q_{r-1}\neq0$ and $q_r=0$. If $q_{r-1}\neq0$, then $q_0,\ldots,q_{r-1}$ are independent and if also $q_r=0$, their invariant span carries a size-$r$ Jordan block at $\lambda$. Thus a preparation may only probe the part of the available Jordan structure that it actually reaches. 

Preparation alone does not ensure experimental resolution because the detector sees only $p_j=Cq_j$. Let $P_r=(p_0,\ldots,p_{r-1})$. The visible depth $d_{\rm vis}^{(z)}$ is one plus the largest index with $p_j\neq0$, or zero for a vanishing response, while the independently resolved depth $d_{\rm rank}^{(z)}$ is the largest $r$ for which $p_0,\ldots,p_{r-1}$ are linearly independent. We thus have the depth hierarchy
\begin{equation}
d_{\rm rank}^{(z)}
\leq
d_{\rm vis}^{(z)}
\leq
d_{\rm act}
\leq
\nu_\lambda.
\label{main_depth_hierarchy}
\end{equation}
The inequalities follow because full column rank requires $p_{r-1}\neq0$, $p_j\neq0$ implies $q_j\neq0$, and $N_\lambda^{\nu_\lambda}=0$ forces $q_{\nu_\lambda}=0$. Thus $p_{r-1}\neq0$ certifies $d_{\rm act}\geq r$ even for a scalar output. Full projected rank is stronger: if $d_{\rm act}\geq r$, then $Q_r=(q_0,\ldots,q_{r-1})$ maps $\mathbb{C}^r$ bijectively onto $\mathcal{K}_r=\operatorname{span}\{q_0,\ldots,q_{r-1}\}$. Since $P_r=CQ_r$, it has rank $r$ exactly when $C$ is injective on $\mathcal{K}_r$, requiring at least $r$ effective output coordinates. This parallels the control-theory distinction between input-generated and output-distinguished directions \cite{Kalman_1960,HoKalman_1966}. A null output remains ambiguous because nonzero $q_r$ may lie in $\ker C$.

Let $\|\cdot\|_x$ and $\|\cdot\|_z$ be fixed norms on $\mathcal{M}$ and $\mathcal{Z}$, respectively. To test true termination, we fix before examining the candidate endpoint a terminal subspace $\mathcal{U}_r^{\rm end}$ containing every possible pair $q_{r-1},q_r$ allowed by the stated uncertainty bounds, and define
\begin{equation}
\underline\alpha_C
=
\inf_{\substack{v\in\mathcal{U}_r^{\rm end}\\\|v\|_x=1}}\|Cv\|_z,
\quad \quad
\overline\beta_C
=
\sup_{\substack{v\in\mathcal{U}_r^{\rm end}\\\|v\|_x=1}}\|Cv\|_z .
\label{main_readout_gains}
\end{equation}
The quantities $\underline\alpha_C$ and $\overline\beta_C$ are, respectively, the weakest and strongest output magnitudes per unit state norm over these allowed endpoint directions. If $\underline\alpha_C>0$, no allowed terminal direction is dark, and $p_{r-1}\neq0$, $p_r=0$ imply $d_{\rm act}=r$.

\textit{Finite-error certification and experimental bounds.}
With finite data, a small fitted coefficient does not establish that it is zero, because nonzero values may remain consistent with its uncertainty. Let $\widehat{p}_j$ denote the reconstructed coefficient and suppose that $\|\widehat{p}_j-p_j\|_z\leq\varepsilon_j$ holds simultaneously for every tested coefficient. At a tested endpoint $r$, $p_{r-1}$ is the last expected nonzero (predecessor) coefficient, whereas $p_r$ is the first coefficient beyond the endpoint (terminal). Then
\begin{equation}
\begin{aligned}
\|\widehat{p}_{r-1}\|_z>\varepsilon_{r-1}
&\Longrightarrow d_{\rm act}\geq r,\\
\|\widehat{p}_r\|_z>\varepsilon_r
&\Longrightarrow d_{\rm act}\geq r+1 .
\end{aligned}
\label{main_nonzero_certificates}
\end{equation}
The second implication certifies nontermination and requires no readout injectivity. Independent resolution requires a matrix test. Let $\widehat{P}_r=(\widehat{p}_0,\cdots,\widehat{p}_{r-1})$. Writing
$\|w\|_z^2=w^\dagger W_zw$, with $W_z=W_z^\dagger>0$, and setting $D_r(\Omega_{\rm sc})=\operatorname{diag}(1,\Omega_{\rm sc}^{-1},\cdots,\Omega_{\rm sc}^{-(r-1)})$ and $\eta_r^2=\sum_{j=0}^{r-1}\Omega_{\rm sc}^{-2j}\varepsilon_j^2$, one has
\begin{equation}
\sigma_{\min}^{\rm col}\!\left(
W_z^{1/2}\widehat{P}_rD_r(\Omega_{\rm sc})
\right)
>\eta_r \quad\Longrightarrow\quad
\operatorname{rank}P_r=r,
\label{main_rank_certificate}
\end{equation}
where $\sigma_{\min}^{\rm col}(M) = \inf_{\|v\|_2=1}\|Mv\|_2$ and the positive reference rate $\Omega_{\rm sc}$ is introduced only to balance the differently scaled columns numerically. Because the same simultaneous event controls every $\Omega_{\rm sc}>0$, it may be optimized after reconstruction while the fit, data, metric (norm), and uncertainty radii remain fixed. The rank test addresses instantaneous measured-direction independence. Bounding termination of the underlying state sequence additionally requires a positive detector-sensitivity bound over the possible endpoint directions and a physical rate scale.

For a linear map $T$ with domain restricted to a subspace $\mathcal{S}$, we write $\|T\|_{a\to b;\mathcal{S}}=\sup_{\substack{v\in\mathcal{S}\\\|v\|_a=1}}\|Tv\|_b$ for its restricted induced operator norm. To quantify whether the Krylov sequence has effectively terminated at order $r$, for $q_{r-1}\neq0$ we define the normalized termination residual $\mathcal{R}_r^{(x)}=\|q_r\|_x/
[\Omega_*\|q_{r-1}\|_x]$, where $\Omega_*=\|N_\lambda\|_{x\to x;\mathcal{G}_\lambda}>0$, and then $0\leq\mathcal{R}_r^{(x)}\leq1$. Let $0<\Omega_K^{\rm L}\leq\Omega_*$ be a calibrated lower rate bound, and suppose the nominal readout satisfies $\widetilde{C}=C+\Delta C$ and $\|\Delta C\|_{x\to z;\mathcal{U}_r^{\rm end}}\leq\delta_C^{\rm end}$, and $\underline\alpha_{\widetilde{C}}$ and $\overline\beta_{\widetilde{C}}$ denote the gains in Eq. (\ref{main_readout_gains}) with $C$ replaced by $\widetilde{C}$. If $\underline\alpha_{\widetilde{C}}>\delta_C^{\rm end}$, we define $\underline\alpha_C^{\rm L}=\underline\alpha_{\widetilde{C}}-\delta_C^{\rm end}>0$ and
$\overline\beta_C^{\rm U}=\overline\beta_{\widetilde{C}}+\delta_C^{\rm end}$. Whenever $\|\widehat{p}_{r-1}\|_z>\varepsilon_{r-1}$, one then has the bound
\begin{equation}
\mathcal{R}_r^{(x)}
\leq
\frac{\overline\beta_C^{\rm U}}
     {\underline\alpha_C^{\rm L}}
\frac{\|\widehat{p}_r\|_z+\varepsilon_r}
{\Omega_K^{\rm L}
 \bigl(\|\widehat{p}_{r-1}\|_z-\varepsilon_{r-1}\bigr)} .
\label{main_closure_bound}
\end{equation}
For a termination tolerance $\tau_{\rm term}$ fixed in advance, we say that the response is terminated to tolerance at order $r$ when $\mathcal{R}_r^{(x)}\leq\tau_{\rm term}$. Thus a right-hand side in Eq. (\ref{main_closure_bound}) no larger than $\tau_{\rm term}$ certifies this statement, although failure gives no conclusion and finite uncertainty cannot prove exact zero.

Before analysis, one specifies the closed sector, candidate exponential rate or its fitting procedure, any additional exponential rates included in the fit, retained polynomial order, effective readout and, for bounding the termination residual, the terminal subspace $\mathcal{U}_r^{\rm end}$ and rate bound $\Omega_K^{\rm L}$. Exponents estimated from the same record require their correlations to be included in the simultaneous coefficient bounds. One then checks the complete sampling matrix and reconstructs $\widehat{p}_j$. If the analysis center is $\widehat{\lambda}=\lambda+\Delta\lambda$, where $\Delta\lambda$ is a centering error, the latter mixes the Krylov vectors of different orders and for an otherwise terminating sequence produces the leading spurious terminal term $q_r^{(\widehat{\lambda})}=-r\Delta\lambda q_{r-1}+O(|\Delta\lambda|^2)$, while fitting over a finite window produces estimator-dependent bias. These contributions must enter every coefficient radius used by the selected certificate.

In real experiments, achieving and stabilizing operation exactly at an EP is very challenging \cite{Bergman_2021,Soleymani_2022,Mortensen_2018}, and in the absence of perfect coalescence, the relevant robustness question is how large the nominally absent next response can get due to a small offset from the EP. As before, we suppose that $N_\lambda^r x_\lambda=0$ at the EP, while $N_\lambda^{r-1}x_\lambda\neq0$. We express the nearby drift and preparation using the same coordinates as at the EP and subtract the average of the eigenvalues that coalesce, so uncertainties in both operations enter the bounds. Thus $B=N_\lambda+E$ and $x=x_\lambda+\eta$, with $\|E\|_{x\to x}\leq\delta_A$ and $\|\eta\|_x\leq\delta_x$. Defining the dimensionless quantities $\epsilon=\delta_A/\|N_\lambda\|_{x\to x}$ and $\zeta=\delta_x/\|x_\lambda\|_x$, the strength of the last nonzero EP response is measured by $\rho_{r-1}=\|N_\lambda^{r-1}x_\lambda\|_x/[(\|N_\lambda\|_{x\to x})^{r-1}\|x_\lambda\|_x]$, and the accumulated error by $\chi_j=(1+\epsilon)^j(1+\zeta)-1$. While $\chi_j$ is the worst-case error accumulated through $j$ applications of the drift, $\rho_{r-1}$ measures how strongly the preparation reaches its last EP response. It directly follows that the ``reopened'' response, quantified by $\mathcal{R}_r(B,x) \equiv \|B^r x\|_x/(\|B\|_{x\to x}\|B^{r-1}x\|_x)$, obeys the bound
\begin{equation}
\begin{aligned}
\mathcal{R}_r(B,x)
\leq
\frac{\chi_r}
{(1-\epsilon)(\rho_{r-1}-\chi_{r-1})},
\end{aligned}
\label{main_near_ep_analytic_bound}
\end{equation}
whenever $\epsilon<1$ and $\chi_{r-1}<\rho_{r-1}$. For small errors, the right-hand side is $(r\epsilon+\zeta)/\rho_{r-1}+O((\epsilon+\zeta)^2)$, so the endpoint remains perturbatively stable, although the guarantee weakens when its last nonzero EP response is small. The same estimates preserve $r$ independent physical response directions whenever the bounded changes are too small to make the exact-EP directions linearly dependent. Away from exact coalescence, this concerns the dimension of the preparation-generated Krylov response, not the Jordan depth of the nearby generator. Complete technical analysis of the certification framework is provided in Supplementary Sec. \ref{supp:krylov_certification_proofs}.

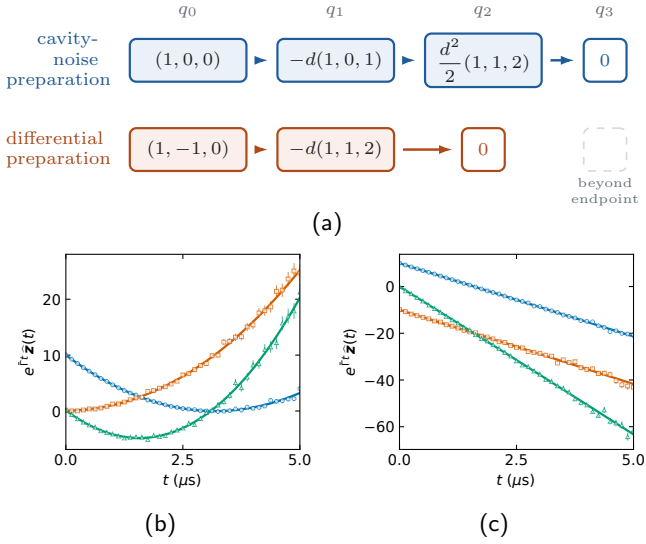
\begin{figure}[!t]
\centering

\resizebox{0.98\columnwidth}{!}{
\begin{tikzpicture}[
    >=Latex,
    font=\sffamily,
    state3/.style={
        draw=KTBlue!90!black,
        fill=KTBlue!9,
        text=KTInk,
        line width=0.85pt,
        rounded corners=2.5pt,
        minimum width=1.58cm,
        minimum height=0.58cm,
        inner sep=2pt,
        align=center,
        font=\scriptsize
    },
    state2/.style={
        draw=KTOrange!90!black,
        fill=KTOrange!9,
        text=KTInk,
        line width=0.85pt,
        rounded corners=2.5pt,
        minimum width=1.58cm,
        minimum height=0.58cm,
        inner sep=2pt,
        align=center,
        font=\scriptsize
    },
    zero3/.style={
        draw=KTBlue!90!black,
        fill=white,
        text=KTBlue!90!black,
        line width=0.85pt,
        rounded corners=2.5pt,
        minimum width=0.58cm,
        minimum height=0.58cm,
        inner sep=1pt,
        align=center,
        font=\scriptsize
    },
    zero2/.style={
        draw=KTOrange!90!black,
        fill=white,
        text=KTOrange!90!black,
        line width=0.85pt,
        rounded corners=2.5pt,
        minimum width=0.58cm,
        minimum height=0.58cm,
        inner sep=1pt,
        align=center,
        font=\scriptsize
    },
    ghost/.style={
        draw=KTRule,
        dashed,
        fill=white,
        line width=0.65pt,
        rounded corners=2.5pt,
        minimum width=0.58cm,
        minimum height=0.58cm,
        inner sep=1pt
    },
    arrow3/.style={
        -{Latex[length=1.9mm,width=1.25mm]},
        draw=KTBlue!90!black,
        line width=0.95pt,
        shorten <=3pt,
        shorten >=3pt
    },
    arrow2/.style={
        -{Latex[length=1.9mm,width=1.25mm]},
        draw=KTOrange!90!black,
        line width=0.95pt,
        shorten <=3pt,
        shorten >=3pt
    },
    row3/.style={
        anchor=east,
        align=right,
        text=KTBlue!85!black,
        font=\scriptsize,
        text width=1.25cm
    },
    row2/.style={
        anchor=east,
        align=right,
        text=KTOrange!85!black,
        font=\scriptsize,
        text width=1.25cm
    },
    order/.style={
        text=KTMuted,
        font=\scriptsize
    },
    endpoint/.style={
        text=KTMuted,
        font=\tiny,
        align=center
    }
]

\path[use as bounding box]
    (-1.35,-1.72) rectangle (7.12,1.08);

\node[order] at (1.00,0.84) {$q_0$};
\node[order] at (3.00,0.84) {$q_1$};
\node[order] at (5.00,0.84) {$q_2$};
\node[order] at (6.65,0.84) {$q_3$};

\node[row3] (lab3) at (-0.10,0.18)
    {cavity-noise\\preparation};

\node[state3] (a0) at (1.00,0.18)
    {$(1,0,0)$};
\node[state3] (a1) at (3.00,0.18)
    {$-d(1,0,1)$};
\node[state3] (a2) at (5.00,0.18)
    {$\dfrac{d^2}{2}(1,1,2)$};
\node[zero3] (a3) at (6.65,0.18)
    {$0$};

\draw[arrow3] (a0.east) -- (a1.west);
\draw[arrow3] (a1.east) -- (a2.west);
\draw[arrow3] (a2.east) -- (a3.west);

\node[row2] (lab2) at (-0.10,-1.03)
    {differential\\preparation};

\node[state2] (b0) at (1.00,-1.03)
    {$(1,-1,0)$};
\node[state2] (b1) at (3.00,-1.03)
    {$-d(1,1,2)$};
\node[zero2] (b2) at (5.00,-1.03)
    {$0$};
\node[ghost] (b3) at (6.65,-1.03)
    {\phantom{$0$}};

\draw[arrow2] (b0.east) -- (b1.west);
\draw[arrow2] (b1.east) -- (b2.west);

\node[endpoint] at (6.65,-1.60)
    {beyond\\ endpoint};

\end{tikzpicture}
}

\par\vspace{0.25mm}
{\sffamily\small (a)\par}

\vspace{1.2mm}

\noindent
\begin{minipage}[t]{0.490\columnwidth}
    \centering
    \includegraphics[width=\linewidth]
    {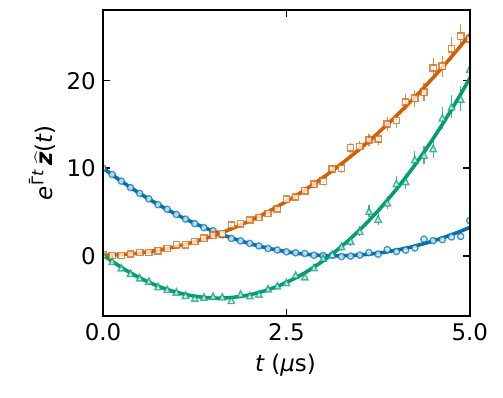}
    \par\vspace{0.25mm}
    {\sffamily\small (b)\par}
\end{minipage}
\hfill
\begin{minipage}[t]{0.490\columnwidth}
    \centering
    \includegraphics[width=\linewidth]
    {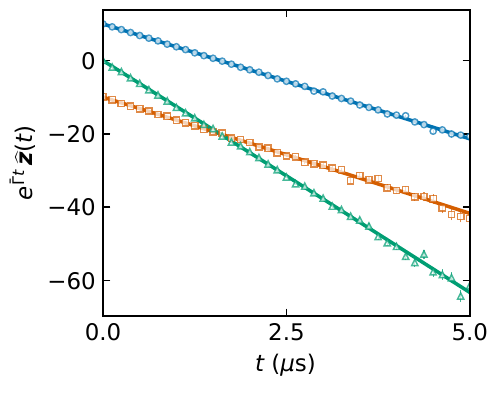}
    \par\vspace{0.25mm}
    {\sffamily\small (c)\par}
\end{minipage}
\caption{\justifying{\textbf{Preparation-dependent Krylov response.} (a) Successive action of $N_{\rm coh}$ [Eq. (\ref{main_optomech_ep3})] generates depth-three and depth-two chains from the cavity-noise and differential preparations under the same size-three Jordan block. Vectors use the basis $\boldsymbol{x}/n_{\rm p}$ with $\boldsymbol{x}=(\delta n_{\rm c},\delta n_{\rm m},\delta Y_{\rm cm})^{T}$,  and $d=(\kappa-\gamma_{\rm m})/2$. (b,c) Demodulated reconstructed outputs $e^{\bar\Gamma t}\widehat{\boldsymbol{z}}(t)$ for the depth-three cavity-noise (b) and depth-two differential (c) preparations, where $\bar\Gamma=(\kappa+\gamma_{\rm m})/2$. Blue circles, orange squares, and green triangles represent $\delta n_{\rm c}$, $\delta n_{\rm m}$, and $\delta Y_{\rm cm}$, respectively. The nominal reconstruction $C=\widetilde{C}=I_3$ gives $\mathbb{E}[\widehat{\boldsymbol{z}}(t)]=\boldsymbol{x}(t)$, so readout uncertainty affects only conservative termination bounds. Solid curves are generalized-least-squares fits using the full noise covariance, cubic in (b) and quadratic in (c); bars are marginal standard deviations of the demodulated records. In (c), occupations are signed differences of independently prepared physical Gaussian trajectories. Parameters are $\omega_{\rm m}/2\pi=1~{\rm MHz}$, $\kappa/2\pi=0.2~{\rm MHz}$, $\gamma_{\rm m}/2\pi=50~{\rm Hz}$, and $T_{\rm m}=4~{\rm K}$, lying within demonstrated resolved-sideband capabilities \cite{GroblacherCooling_2009,GroblacherStrongCoupling_2009,Weaver_2017}; the EP condition fixes $G=G_{\rm EP}=(\kappa-\gamma_{\rm m})/4$, and each physical run adds $n_{\rm p}=10$.}}
\label{fig3}
\end{figure}

\textit{Red-sideband optomechanical illustration.}
Optomechanics provides a concrete illustration of our formalism because in that case bilinear coherences can lift a field-amplitude EP2 into moment-space Jordan dynamics accessible through occupation-number measurements. We use the resonant red-sideband rotating-wave model \cite{Aspelmeyer_2014}, which retains the resonant excitation-exchange interaction $H_{\rm int}=\hbar G(a^\dagger b+ab^\dagger)$, where $a$ and $b$ are the annihilation operators for the cavity and mechanical fluctuations, respectively, and $G$ is the linearized optomechanical coupling. The cavity and mechanical damping rates are $\kappa$ and $\gamma_{\rm m}$, respectively. The field-amplitude EP2 is at $G_{\rm EP}=(\kappa-\gamma_{\rm m})/4$ for
$\kappa>\gamma_{\rm m}$ \cite{Ghosh_2026}. After subtracting the stationary Gaussian background, we define the excess photon and phonon occupations as $\delta n_{\rm c}=\delta\langle a^\dagger a\rangle$ and $\delta n_{\rm m}=\delta\langle b^\dagger b\rangle$, together with the cavity-mechanics coherence $\delta Y_{\rm cm}=2{\rm Im}[\delta\langle a^\dagger b\rangle]$. Writing the moment vector $\boldsymbol{x}=(\delta n_{\rm c},\delta n_{\rm m}, \delta Y_{\rm cm})^T$ and defining $\bar\Gamma=(\kappa+\gamma_{\rm m})/2$, at this EP, one has
\begin{equation}
\dot{\boldsymbol{x}} = \left(-\bar\Gamma I+N_{\rm coh}\right)\boldsymbol{x}, \quad
N_{\rm coh} =
G_{\rm EP}
\begin{pmatrix}
-2&0&1\\
0&2&-1\\
-2&2&0
\end{pmatrix},
\label{main_optomech_ep3}
\end{equation}
where $N_{\rm coh}^{3}=0$ and $N_{\rm coh}^{2}\neq0$, meaning that the invariant three-coordinate coherence sector carries a size-three Jordan block. The remaining $2{\rm Re}[\delta\langle a^\dagger b\rangle]$ forms a decoupled size-one block, giving a $3\oplus1$ partition of the coherence space formed by bilinears that conserve the total photon-plus-phonon excitation number. For $\boldsymbol{v}=(v_{\rm c},v_{\rm m},v_Y)^T$, we equip this three-coordinate sector with the norm $\|\boldsymbol{v}\|_{\rm coh}^2
=2|v_{\rm c}|^2+2|v_{\rm m}|^2+|v_Y|^2$. The corresponding induced norm of the shifted generator is $\Omega_*=\|N_{\rm coh}\|_{\rm coh\to coh}=(\kappa-\gamma_{\rm m})/\sqrt{2}$.

Preparation determines how much of this structure is activated. In the ideal limit of a pulse short enough that the system evolves negligibly during it, a cavity-noise preparation that initially increases only the cavity occupation realizes $\boldsymbol{x}^{(3)}_0=n_{\rm p}(1,0,0)^T$ and activates three successive response directions [Fig. \ref{fig3}(a)]. By contrast, subtracting the measured evolution of a run with an added mechanical occupation from that of an otherwise identical run with the same added cavity occupation defines the differential preparation $\boldsymbol{x}^{(2)}_0=n_{\rm p}(1,-1,0)^T$, which activates two directions [Fig. \ref{fig3}(a)]. Here $n_{\rm p}>0$ is the occupation added in each physical run, and the negative component of $\boldsymbol{x}^{(2)}_0$ arises from subtracting the two records. The latter preparation is implemented as the difference of two independently prepared physical Gaussian evolutions so that the covariance contributions from the two measured records add. Finite-duration preparation errors must be included in the preparation-error bound. Defining $d=(\kappa-\gamma_{\rm m})/2=2G_{\rm EP}$, the cavity responses are
\begin{equation}
\begin{aligned}
\delta n_{\rm c}^{(3)}(t)
&=
n_{\rm p}e^{-\bar\Gamma t}
\left(1-dt+\frac{d^2t^2}{4}\right),
\\
\delta n_{\rm c}^{(2)}(t)
&=
n_{\rm p}e^{-\bar\Gamma t}(1-dt).
\end{aligned}
\label{main_optomech_traces}
\end{equation}
One finite-imprecision simulation is shown in Fig. \ref{fig3}(b,c). Cubic and quadratic fits retain the first coefficient beyond their exact endpoints. A scalar cavity-output map is noninjective on the three-dimensional terminal sector and cannot bound its state-space termination residual, so we use a proposed three-output model: the cavity and mechanical populations, plus their coherence obtained by differencing two simulated occupation channels with opposite coherence signs. The nominal map is $C=\widetilde{C}=I_3$. Predeclared deterministic $2\%$ calibration margins on the readout map and characteristic rate, fixed independently of the record, give $\|\Delta C\|_{\rm coh\to coh}\leq0.02$, $\underline\alpha_C^{\rm L}=0.98$, $\overline\beta_C^{\rm U}=1.02$, and $\Omega_K^{\rm L}=0.98\Omega_*$; statistical fluctuations enter only the coefficient covariance. Separate joint $95\%$ regions covering all fitted coefficients give the predecessor and terminal ratios $\|\widehat{p}_2\|_z/\varepsilon_2>4.51$ and $\|\widehat{p}_3\|_z/\varepsilon_3<0.59$ for depth three, and $\|\widehat{p}_1\|_z/\varepsilon_1>31.9$ and $\|\widehat{p}_2\|_z/\varepsilon_2<0.25$ for depth two. The corresponding projected-rank tests certify $\operatorname{rank}P_3=3$ and $\operatorname{rank}P_2=2$, but no $95\%$ coverage is claimed for the intersection of these preparation-specific events. Let $B_r$ denote the right-hand side of Eq. (\ref{main_closure_bound}). Although $N_{\rm coh}^3=0$ and $N_{\rm coh}^2\boldsymbol{x}^{(2)}_0=0$ make both terminal vectors vanish analytically at the exact EP, applying the bound without these identities gives $B_3\simeq0.313495$ and $B_2\simeq0.031269$. Thus this finite-error test meets the stated $10\%$ tolerance only for depth two, and $B_3>0.10$ reflects limited finite-data precision. The noisy data certify both nonzero predecessors and projected ranks, while the exact identities fix both endpoints. These bounds assume the exact-EP generator and stated preparations.

The near-EP framework applies elegantly to the optomechanical example. Within the resonant red-sideband rotating-wave family, let us consider the differential preparation that activates two independent directions. For variations of $d$ and $G$ within this family, departure from the EP does not create a third independent Krylov direction, because the next vector returns along the initially prepared direction. With $d=(\kappa-\gamma_{\rm m})/2$ and taking $G>0$ by phase convention, direct multiplication gives
\begin{equation}
\boldsymbol{q}_2^{(2)} =(d^2-4G^2)\boldsymbol{q}_0^{(2)}, \quad \quad 
\mathcal{R}_2^{(2)}
=
\frac{|d^2-4G^2|}{d^2+4G^2}.
\label{main_optomech_near_ep_residual}
\end{equation}
The two amplitude eigenvalues obey $|\lambda_+-\lambda_-|^2=|d^2-4G^2|$, so the endpoint reopening is exactly the squared mode splitting normalized by $d^2+4G^2$. For $0\leq b<1$, the condition $\mathcal{R}_2^{(2)}\leq b$ is equivalent to the parameter window $\sqrt{(1-b)/(1+b)}\leq4G/(\kappa-\gamma_{\rm m})\leq\sqrt{(1+b)/(1-b)}$. By contrast, the cavity-noise preparation spans three independent response directions for every $G\neq0$. Robustness near the same EP therefore depends on preparation, just as the exact-EP depth does. Details of the optomechanical model calculations and simulations are given in Supplementary Sec. \ref{supp:optomech_realization}.

\textit{Conclusion.} We have shown that Krylov tomography separates generator-available, preparation-activated, visible, and independently resolved exceptional dynamics. At an EP, coefficients certified as nonzero bound activated depth, projected-rank tests certify independent resolution, and restricted readout gain distinguishes an endpoint from detector blindness; finite uncertainty yields preparation-specific upper bounds on the termination residual. Small bounded departures from the EP need not invalidate the construction, as the general bound controls endpoint reopening. We have illustrated our theoretical proposal using the physically prominent example of resonant red-sideband optomechanics. Krylov tomography thus turns EP structure of a physical system into a quantitative account of what an experiment actually prepares, resolves, and rules out, with controlled uncertainty.

\noindent
\textit{Acknowledgements.} A.G. and M.B. thank the Air Force Office of Scientific Research (FA9550-23-1-0259) for support.

\noindent
\textit{Data availability.}
Source codes to reproduce all the simulation results are available at \url{https://github.com/aritraghoshphysics/Krylov-Tomography-of-EP-Dynamics}.

\clearpage

\setcounter{section}{0}
\setcounter{equation}{0}
\setcounter{figure}{0}
\setcounter{table}{0}
\setcounter{page}{1}

\renewcommand{\thesection}{S\arabic{section}}
\renewcommand{\theequation}{S\arabic{equation}}
\renewcommand{\thefigure}{S\arabic{figure}}
\renewcommand{\thetable}{S\arabic{table}}
\renewcommand{\thepage}{S\arabic{page}}

\onecolumngrid 
\begin{center}
    \textbf{\large Supplementary Material for: Krylov Tomography and\\ Finite-Uncertainty Certification of Exceptional-Point Dynamics} \\
    \vspace{4mm}
    Aritra Ghosh and M. Bhattacharya,\\
    {\it School of Physics and Astronomy, Rochester Institute of Technology,\\ 84 Lomb Memorial Drive, Rochester, New York 14623, USA}
\end{center}

\section{Krylov tomography and certification}
\label{supp:krylov_certification_proofs}

We now describe the finite-error certification theory in detail. All spaces are finite dimensional over $\mathbb{C}$, so real moment systems are understood through complexification. If $A$ is real and $\lambda$ is nonreal, its contribution must be combined with the conjugate sector to obtain the real signal $2{\rm Re}[e^{\lambda t}\sum_jt^jq_j/j!]$. We use the following assumptions:
\begin{enumerate}
\item[(i)] After subtraction of any stationary part, the retained moments form a declared closed sector $\mathcal{M}$ and obey $\dot x=Ax$, with $A$ constant during the measurement record.
\item[(ii)] The candidate eigenvalue $\lambda$ is isolated within $A|_{\mathcal{M}}$, and its contribution is separated by an independently validated spectral or exponential-polynomial reconstruction. A symmetry-resolved invariant sector may likewise be selected only through independent validation.
\item[(iii)] The instantaneous readout is a fixed linear map $C:\mathcal{M}\to\mathcal{Z}$ into a finite-dimensional output space.
\item[(iv)] Exact statements are pointwise in the physical $A$, $\lambda$, $x_\lambda$, and $C$. Every finite-error statement is conditioned on one deterministic uncertainty set or joint confidence event containing all coefficient, readout, centering, preparation, and rate bounds invoked by that statement.
\end{enumerate}
No statistical independence among contributions to a simultaneous bound is assumed. The near-EP construction instead compares a calibrated family with the exact-EP point after identifying the corresponding isolated sectors.

These assumptions separate three logically distinct inputs. Closure identifies which moments obey a finite linear model; spectral isolation identifies which part of that model is being tested; and readout calibration determines what can be inferred about an unobserved state vector from its measured image. None follows from the other two. In particular, observing a polynomially modified exponential does not by itself validate the chosen closed sector or exclude unresolved neighboring exponents.

\subsection{Isolation and coefficient identifiability}
\label{supp:sector_isolation}

For a contour $\Gamma_\lambda$ enclosing no other eigenvalue of $A|_{\mathcal{M}}$, the Riesz projector is
\begin{equation}
\Pi_\lambda=\frac{1}{2\pi i}\oint_{\Gamma_\lambda}(zI-A)^{-1}dz,
\quad \quad
\mathcal{G}_\lambda=\operatorname{ran}\Pi_\lambda .
\label{supp_riesz_projection}
\end{equation}
To verify Eq. (\ref{supp_riesz_projection}), we choose $A=SJS^{-1}$. On a size-$m_\mu$ Jordan block $J_\mu=\mu I+N_\mu$, one has
\begin{equation}
(zI-J_\mu)^{-1}=\sum_{k=0}^{m_\mu-1}\frac{N_\mu^k}{(z-\mu)^{k+1}},
\end{equation}
so for $\mu\neq\lambda$ every term is analytic inside $\Gamma_\lambda$ and integrates to zero. For $\mu=\lambda$, only the $k=0$ term has a nonzero residue, equal to the identity. Thus $\Pi_\lambda$ is the identity on all blocks at $\lambda$ and zero on every other block. Consequently $\Pi_\lambda^2=\Pi_\lambda$, $[A,\Pi_\lambda]=0$, and its range is precisely $\mathcal{G}_\lambda$ \cite{KatoSupp,HornJohnsonSupp}.

The complementary projector $I-\Pi_\lambda$ is also invariant, so every preparation decomposes uniquely as $x_0=x_\lambda+x_\perp$, with $x_\lambda=\Pi_\lambda x_0$ and $x_\perp=(I-\Pi_\lambda)x_0$. All our depth statements concern only $x_\lambda$ and require $x_\lambda\neq0$, so contributions from $x_\perp$ must be removed or included explicitly in the reconstruction.

Now we define $x_\lambda=\Pi_\lambda x_0\neq0$, $N_\lambda=(A-\lambda I)|_{\mathcal{G}_\lambda}$, $\nu_\lambda=\min\{\nu\geq1:N_\lambda^\nu=0\}$, $q_j=N_\lambda^jx_\lambda$, and $p_j=Cq_j$. Since $A|_{\mathcal{G}_\lambda}=\lambda I+N_\lambda$ and the two terms commute, nilpotency gives
\begin{equation}
e^{At}x_\lambda=e^{\lambda t}\sum_{j=0}^{\nu_\lambda-1}\frac{t^j}{j!}q_j,
\label{supp_isolated_response}
\end{equation}
so applying $C$ gives
\begin{equation}
z_\lambda(t)=e^{\lambda t}\sum_{j=0}^{\nu_\lambda-1}\frac{t^j}{j!}p_j .
\label{supp_projected_response}
\end{equation}
Thus $q_j$ and $p_j$ are the Taylor coefficients at $t=0$ of the demodulated state and output, respectively. An exponential-polynomial fit directly reconstructs $p_j$, and the reconstruction of $q_j$ additionally requires state access or an injective calibrated readout.

At sampling times $t_\ell$, let $V_{\ell j}=t_\ell^j/j!$. Eq. (\ref{supp_projected_response}) becomes
\begin{equation}
\operatorname{col}_{\ell}\!\left[e^{-\lambda t_\ell}z_\lambda(t_\ell)\right]
=(V\otimes I_{\mathcal{Z}})\operatorname{col}(p_0,\ldots,p_{\nu_\lambda-1}).
\label{supp_sampling_system}
\end{equation}
If at least $\nu_\lambda$ times are distinct, $V$ has full column rank: $Va=0$ would define a polynomial $\sum_ja_jt^j/j!$ of degree at most $\nu_\lambda-1$ with at least $\nu_\lambda$ roots, forcing $a=0$. Hence the coefficients are unique. This is an exact identifiability statement, not a stability statement. With noise, the uncertainty of the recovered coefficients depends on the smallest singular value of the weighted sampling matrix and can become large for a short or poorly spaced time window even though the coefficients remain unique. Increasing the number of output coordinates does not repair a rank-deficient temporal design when the coefficient vectors are otherwise unrestricted.

More generally, jointly fitted fixed exponents $\mu_a$ and orders $m_a$ give the confluent sampling matrix $H_{\ell,(a,j)}=e^{\mu_at_\ell}t_\ell^j/j!$, and the exact coefficients are identifiable if and only if $H$ has full column rank and its conditioning enters their uncertainty. This statement is conditional on fixed exponents and exponents estimated from the same record require their nonlinear correlations, and any data-dependent model selection, to be covered by the simultaneous coefficient event.

Finally, we point out that sector isolation is necessary for finite-power termination. On $\mathcal{G}_\mu$ with $\mu\neq\lambda$, $(A-\lambda I)=(\mu-\lambda)I+N_\mu$ is invertible, with
\begin{equation}
[(\mu-\lambda)I+N_\mu]^{-1}
=\frac{1}{\mu-\lambda}\sum_{k=0}^{m_\mu-1}
\left(-\frac{N_\mu}{\mu-\lambda}\right)^k.
\end{equation}
No nonzero component outside $\mathcal{G}_\lambda$ can therefore be annihilated by a finite power of $A-\lambda I$. Thus a terminating sequence cannot be assigned to $\lambda$ until spectral leakage has been excluded or bounded. In finite error, any imperfect subtraction of other sectors enters the coefficient radii rather than the exact algebraic proof.

\subsection{Activated depth, readout, and true termination}
\label{supp:exact_activation}
\label{supp:readout}

Let $d_{\rm act}=\min\{r\geq1:q_r=0\}$, and the set is nonempty because $N_\lambda$ is nilpotent. The following result establishes the structure selected by the preparation, without claiming the full Jordan partition of $\mathcal{G}_\lambda$: 

\vspace{1mm}

\begin{prop}\label{prop1}
For every $r\geq1$,
\begin{equation}
d_{\rm act}=r\quad\Longleftrightarrow\quad q_{r-1}\neq0\ \text{and}\ q_r=0.
\label{supp_depth_equivalence}
\end{equation}
Whenever $q_{r-1}\neq0$, the vectors $q_0,\ldots,q_{r-1}$ are independent; if also $q_r=0$, their invariant span carries one size-$r$ Jordan block at $\lambda$.
\end{prop}

\textbf{Proof.} If $d_{\rm act}=r$, minimality gives $q_{r-1}\neq0$ and $q_r=0$. Conversely, $q_j=0$ for any $j<r$ would imply $q_{r-1}=N_\lambda^{r-1-j}q_j=0$, so the two stated conditions make $r$ the first terminating power. For independence, let us suppose
\begin{equation}
\sum_{j=0}^{r-1}c_jq_j=0
\label{supp_linear_relation}
\end{equation}
and let $\ell$ be the smallest index with $c_\ell\neq0$. Writing $d=d_{\rm act}\geq r$ and applying $N_\lambda^{d-1-\ell}$ leaves $c_\ell q_{d-1}=0$: all terms with $j>\ell$ contain a power at least $d$ and vanish. This contradicts $q_{d-1}\neq0$. Finally, $N_\lambda q_j=q_{j+1}$ and $q_r=0$ make the span invariant under both $N_\lambda$ and $A$. In the reversed basis $(q_{r-1},\ldots,q_0)$, the restriction is the standard nilpotent shift plus $\lambda I$. \hfill$\square$

Consequently,
\begin{equation}
q_{r-1}\neq0\quad\Longleftrightarrow\quad d_{\rm act}\geq r .
\label{supp_depth_lower_bound}
\end{equation}
The number $d_{\rm act}$ is the length of the single cyclic chain reached from the chosen preparation. It may be smaller than $\nu_\lambda$, and neither its value nor termination of that chain determines how many other Jordan blocks occur at $\lambda$. Changing the preparation can therefore change $d_{\rm act}$ without changing the generator.

Defining $Q_r=(q_0,\ldots,q_{r-1})$, $P_r=CQ_r=(p_0,\ldots,p_{r-1})$, and $\mathcal{K}_r=\operatorname{ran}Q_r$, the following is true: 

\vspace{1mm}

\begin{prop}\label{prop2}
$\operatorname{rank}P_r\leq\operatorname{rank}Q_r$, and $p_{r-1}\neq0$ implies $d_{\rm act}\geq r$. If $d_{\rm act}\geq r$, then
\begin{equation}
\operatorname{rank}P_r=r\quad\Longleftrightarrow\quad\ker C\cap\mathcal{K}_r=\{0\}.
\label{supp_rank_injectivity}
\end{equation}
\end{prop}
\textbf{Proof.} The first statement follows from $P_r=CQ_r$, because a linear map cannot increase the dimension of a column space. The second follows from $p_{r-1}\neq0\Rightarrow q_{r-1}\neq0$ and Eq. (\ref{supp_depth_lower_bound}). Under $d_{\rm act}\geq r$, Proposition \ref{prop1} makes $Q_r:\mathbb{C}^r\to\mathcal{K}_r$ an isomorphism. Therefore $P_ra=0$ is equivalent to $Q_ra\in\ker C\cap\mathcal{K}_r$. Either side of Eq. (\ref{supp_rank_injectivity}) then implies the other. \hfill$\square$

We define $d_{\rm vis}^{(z)}=0$ if every $p_j$ vanishes and otherwise $d_{\rm vis}^{(z)}=1+\max\{j:p_j\neq0\}$, and let $d_{\rm rank}^{(z)}=\max\{r:\operatorname{rank}P_r=r\}$, with the maximum of the empty set equal to zero. Full rank forces $p_{r-1}\neq0$; a nonzero $p_j=Cq_j$ forces $q_j\neq0$; and nilpotency gives $d_{\rm act}\leq\nu_\lambda$. Hence
\begin{equation}
d_{\rm rank}^{(z)}\leq d_{\rm vis}^{(z)}\leq d_{\rm act}\leq\nu_\lambda .
\label{supp_depth_hierarchy}
\end{equation}
Each inequality can be strict. Detector null directions can reduce visible depth below activated depth, and several visible coefficient vectors can be linearly dependent, reducing projected rank further. Conversely, a nonzero scalar coefficient can certify a lower bound on activation even though a scalar instantaneous output can never have projected rank exceeding one. This is why visibility and independent resolution must be reported separately.

A measured zero alone is inconclusive. For example, take $Ne_4=e_3$, $Ne_3=e_2$, $Ne_2=e_1$, $Ne_1=0$, prepare $e_4$, and choose $Ce_4=(1,0)^T$, $Ce_3=(0,1)^T$, $Ce_2=0$, and $Ce_1=(1,0)^T$. Then $p_2=0$ although $q_2\neq0$ and the later coefficient $p_3$ reappears. We therefore fix, before examining terminal coefficients, a nonzero subspace $\mathcal{U}_r^{\rm end}$ satisfying
\begin{equation}
q_{r-1},q_r\in\mathcal{U}_r^{\rm end}
\label{supp_terminal_subspace}
\end{equation}
for every admissible calibration. With fixed norms on $\mathcal{M}$ and $\mathcal{Z}$, we define
\begin{equation}
\underline\alpha_C=\inf_{\substack{v\in\mathcal{U}_r^{\rm end}\\\|v\|_x=1}}\|Cv\|_z,
\quad \quad
\overline\beta_C=\sup_{\substack{v\in\mathcal{U}_r^{\rm end}\\\|v\|_x=1}}\|Cv\|_z .
\label{supp_restricted_gains}
\end{equation}
Compactness of the unit sphere gives
\begin{equation}
\underline\alpha_C>0\quad\Longleftrightarrow\quad\ker C\cap\mathcal{U}_r^{\rm end}=\{0\},
\label{supp_gain_injectivity}
\end{equation}
since a nonzero vector in the intersection produces a zero on the unit sphere, and if the intersection is trivial, the continuous positive function $v\mapsto\|Cv\|_z$ has a positive minimum on that compact sphere.

The containment in Eq. (\ref{supp_terminal_subspace}) is part of the calibration claim: if generator or preparation uncertainty changes the possible endpoint directions, $\mathcal{U}_r^{\rm end}$ must contain their union. Choosing it after viewing the terminal estimate would invalidate the stated coverage. The upper gain is always finite in finite dimension, while the strictly positive lower gain is the additional ingredient that rules out detector blindness.

\vspace{1mm}

\begin{prop}\label{prop3}
If $\underline\alpha_C>0$, then $p_r=0$ if and only if $q_r=0$; hence $p_{r-1}\neq0$ and $p_r=0$ imply $d_{\rm act}=r$. If also $\operatorname{rank}P_r=r$, then $C$ is injective on all of $\mathcal{K}_r$.
\end{prop}
\textbf{Proof.} The implication $q_r=0\Rightarrow p_r=0$ is immediate. Conversely, $p_r=0$ and Eq. (\ref{supp_terminal_subspace}) place $q_r$ in $\ker C\cap\mathcal{U}_r^{\rm end}$, which is trivial by Eq. (\ref{supp_gain_injectivity}). Moreover, $p_{r-1}\neq0$ implies $q_{r-1}\neq0$, so Proposition \ref{prop1} gives $d_{\rm act}=r$; the last claim is Proposition \ref{prop2}. \hfill$\square$

Thus direct resolution of $r$ reached directions requires at least $r$ independent instantaneous output coordinates or calibrated branches. This dimension count concerns the instantaneous map $C$ and does not exclude dynamic observability of a lower-dimensional time record when a calibrated drift is used. More precisely, Eq. (\ref{supp_rank_injectivity}) requires $\dim\mathcal{Z}\geq r$, while endpoint injectivity requires $\dim\mathcal{Z}\geq\dim\mathcal{U}_r^{\rm end}$. Outputs obtained in separate experimental branches may be stacked into $\mathcal{Z}$ only when their relative gains and coordinate conventions are calibrated, so repeated time samples of one scalar channel cannot by themselves increase the instantaneous projected rank.

\subsection{Finite-error certification}
\label{supp:finite_rank}

On one simultaneous event $\mathcal{E}$, let us assume
\begin{equation}
\|\widehat{p}_j-p_j\|_z\leq\varepsilon_j
\label{supp_coefficient_bounds}
\end{equation}
for every coefficient used by the selected certificate. The radii may combine statistical fitting error with bounded isolation, centering, finite-window, preparation, and coordinate-reconstruction biases. Deterministic uncertainty in the physical mean readout is treated through the gain bounds below and is not counted again in $\varepsilon_j$; no independence among contributions to the joint event is required. Now the reverse triangle inequality gives
\begin{equation}
\|\widehat{p}_j\|_z>\varepsilon_j
\quad\Longrightarrow\quad p_j\neq0
\quad\Longrightarrow\quad q_j\neq0
\quad\Longrightarrow\quad d_{\rm act}\geq j+1 .
\label{supp_nonzero_certificate}
\end{equation}
The implication is one-sided and uses a strict inequality. If $\|\widehat{p}_j\|_z\leq\varepsilon_j$, the coefficient is uncertified: it is neither certified as zero nor as nonzero. The event may be a deterministic uncertainty set, a simultaneous confidence region, or projections of a joint ellipsoid, provided that all coefficients entering the same conclusion are covered together.

Coverage must match the reported statement. A depth claim at one order needs the corresponding coefficient bound; a rank claim needs simultaneous control of every column in $P_r$; and the closure bound needs the predecessor and terminal bounds together with the readout and rate calibration bounds. Separate confidence events may be quoted separately, but their intersection cannot retain the same nominal coverage without an additional joint construction.

Let us denote $\|w\|_z^2=w^\dagger W_zw$ with $W_z>0$, and let $\sigma_{\min}^{\rm col}(M)=\inf_{\|a\|_2=1}\|Ma\|_2$, which is positive exactly when $M$ has full column rank. Upon defining $D_r(\Omega)=\operatorname{diag}(1,\Omega^{-1},\ldots,\Omega^{-(r-1)})$ and $\eta_r(\Omega)^2=\sum_{j=0}^{r-1}\Omega^{-2j}\varepsilon_j^2$, the following holds:

\vspace{1mm}

\begin{prop}\label{prop4}
On $\mathcal{E}$, for every $\Omega>0$,
\begin{equation}
\left\|W_z^{1/2}(\widehat{P}_r-P_r)D_r(\Omega)\right\|_2\leq\eta_r(\Omega).
\label{supp_uniform_scaled_error}
\end{equation}
Consequently, for any $\Omega_{\rm sc}>0$,
\begin{equation}
\sigma_{\min}^{\rm col}\!\left(W_z^{1/2}\widehat{P}_rD_r(\Omega_{\rm sc})\right)>\eta_r(\Omega_{\rm sc})
\quad\Longrightarrow\quad\operatorname{rank}P_r=r .
\label{supp_rank_margin}
\end{equation}
\end{prop}
\textbf{Proof.} We set $E_r=\widehat{P}_r-P_r$, so Eq. (\ref{supp_coefficient_bounds}) yields
\begin{equation}
\|W_z^{1/2}E_rD_r(\Omega)\|_2^2
\leq\|W_z^{1/2}E_rD_r(\Omega)\|_F^2
=\sum_{j=0}^{r-1}\Omega^{-2j}\|\widehat{p}_j-p_j\|_z^2
\leq\eta_r(\Omega)^2 .
\end{equation}
For $\widehat{M}=W_z^{1/2}\widehat{P}_rD_r$ and $M=W_z^{1/2}P_rD_r$, the reverse triangle inequality, followed by an infimum over unit $a$, gives $\sigma_{\min}^{\rm col}(M)\geq\sigma_{\min}^{\rm col}(\widehat{M})-\|\widehat{M}-M\|_2$. Eq. (\ref{supp_rank_margin}) therefore makes $M$ full column rank. Since $W_z^{1/2}$ and $D_r$ are invertible, $P_r$ also has rank $r$. \hfill$\square$

The error bound is simultaneous for every $\Omega>0$ on the same event, so $\Omega_{\rm sc}$ may be optimized after reconstruction if the fitted coefficients, metric, and radii are held fixed. By contrast, data-dependent selection of the sector, polynomial order, endpoint, sampling subset, metric, regularization, or fit must be covered simultaneously or performed on independent data.

Notably, the diagonal scaling only compares columns with different physical units, so it neither changes their span nor supplies dynamical information. The rank certificate is sufficient rather than necessary: failure of the singular-value margin leaves the true rank unresolved. It also certifies the physical matrix $P_r$, not merely the rank of the noisy estimate $\widehat{P}_r$.

\subsection{Readout calibration and termination within tolerance}
\label{supp:readout_transfer}

For every $v\in\mathcal{U}_r^{\rm end}$, the definitions imply
\begin{equation}
\underline\alpha_C\|v\|_x\leq\|Cv\|_z\leq\overline\beta_C\|v\|_x,
\label{supp_gain_inequality}
\end{equation}
and hence
\begin{equation}
\|q_r\|_x\leq\frac{\|p_r\|_z}{\underline\alpha_C},
\quad \quad
\|q_{r-1}\|_x\geq\frac{\|p_{r-1}\|_z}{\overline\beta_C}.
\label{supp_state_output_transfer}
\end{equation}
Let the nominal map be $\widetilde{C}=C+\Delta C$, with
\begin{equation}
\|\Delta C\|_{x\to z;\mathcal{U}_r^{\rm end}}\leq\delta_C^{\rm end}.
\label{supp_readout_error_bound}
\end{equation}
For any unit $v\in\mathcal{U}_r^{\rm end}$, $\|Cv\|_z\geq\|\widetilde{C}v\|_z-\|\Delta Cv\|_z$ and $\|Cv\|_z\leq\|\widetilde{C}v\|_z+\|\Delta Cv\|_z$. Taking the infimum and supremum gives
\begin{equation}
\underline\alpha_C\geq\underline\alpha_{\widetilde{C}}-\delta_C^{\rm end}.
\label{supp_lower_gain_bound}
\end{equation}
Likewise,
\begin{equation}
\overline\beta_C\leq\overline\beta_{\widetilde{C}}+\delta_C^{\rm end}.
\label{supp_upper_gain_bound}
\end{equation}
When $\underline\alpha_{\widetilde{C}}>\delta_C^{\rm end}$, let us define
\begin{equation}
\underline\alpha_C\geq\underline\alpha_{\widetilde{C}}-\delta_C^{\rm end}
\equiv\underline\alpha_C^{\rm L}>0,
\quad \quad
\overline\beta_C\leq\overline\beta_{\widetilde{C}}+\delta_C^{\rm end}
\equiv\overline\beta_C^{\rm U}.
\label{supp_conservative_gains}
\end{equation}
The first bound remains strictly positive uniformly over the declared calibration class, so every admissible physical readout is injective on the endpoint subspace. If it is nonpositive, a dark endpoint remains allowed and no state-level upper bound can be obtained from the measured terminal coefficient alone.

Now let
\begin{equation}
\Omega_*=\|N_\lambda\|_{x\to x;\mathcal{G}_\lambda}>0.
\label{supp_krylov_rate}
\end{equation}
For $q_{r-1}\neq0$, we define
\begin{equation}
\mathcal{R}_r^{(x)}=\frac{\|q_r\|_x}{\Omega_*\|q_{r-1}\|_x}.
\label{supp_closure_residual}
\end{equation}
Since $q_r=N_\lambda q_{r-1}$, one must have $0\leq\mathcal{R}_r^{(x)}\leq1$. Independent calibration must provide
\begin{equation}
0<\Omega_K^{\rm L}\leq\Omega_*.
\label{supp_krylov_rate_lower}
\end{equation}
For example, $\|\widetilde{N}_\lambda-N_\lambda\|_{x\to x}\leq\delta_N$ permits $\Omega_K^{\rm L}=\|\widetilde{N}_\lambda\|_{x\to x}-\delta_N>0$ by the reverse triangle inequality. This is a lower bound on the physical shifted-generator norm, not an assertion that an arbitrary perturbation preserves the EP Jordan structure.

The normalization by $\Omega_*$ makes $\mathcal{R}_r^{(x)}$ dimensionless and bounded by unity. It compares the reopening at level $r$ with the largest calibrated action of the shifted generator in the same norm. So it is not a probability and should not be compared across different physical metrics without transporting the metric. The following result is true:

\vspace{1mm}

\begin{prop}\label{prop5}
If $\|\widehat{p}_{r-1}\|_z>\varepsilon_{r-1}$, then
\begin{equation}
\mathcal{R}_r^{(x)}\leq
\frac{\overline\beta_C^{\rm U}}{\underline\alpha_C^{\rm L}}
\frac{\|\widehat{p}_r\|_z+\varepsilon_r}
{\Omega_K^{\rm L}(\|\widehat{p}_{r-1}\|_z-\varepsilon_{r-1})}.
\label{supp_closure_bound}
\end{equation}
\end{prop}
\textbf{Proof.} The predecessor condition and Eq. (\ref{supp_nonzero_certificate}) make $q_{r-1}\neq0$. Eqs. (\ref{supp_state_output_transfer}) and (\ref{supp_conservative_gains}), followed by Eq. (\ref{supp_krylov_rate_lower}), give
\begin{equation}
\mathcal{R}_r^{(x)}
\leq\frac{\overline\beta_C^{\rm U}}{\underline\alpha_C^{\rm L}}
\frac{\|p_r\|_z}{\Omega_K^{\rm L}\|p_{r-1}\|_z}.
\label{supp_intermediate_closure}
\end{equation}
Finally, Eq. (\ref{supp_coefficient_bounds}) gives $\|p_r\|_z\leq\|\widehat{p}_r\|_z+\varepsilon_r$ and $\|p_{r-1}\|_z\geq\|\widehat{p}_{r-1}\|_z-\varepsilon_{r-1}>0$. Substitution proves Eq. (\ref{supp_closure_bound}). \hfill$\square$

For a tolerance $\tau_{\rm term}$ fixed before inspecting the terminal coefficient, a right-hand side not exceeding $\tau_{\rm term}$ certifies termination within tolerance. A positive-width uncertainty set generally cannot establish the exact equality $q_r=0$, and failure of the inequality does not establish nontermination. The statement is preparation-specific and refers to the declared moment norm.

The strict predecessor test is essential, as without it, the denominator can vanish and the normalized residual is not operationally bounded. Conversely, a terminal coefficient certified as nonzero proves nontermination at that level through Eq. (\ref{supp_nonzero_certificate}), so an uncertified terminal coefficient is used only through its upper error bar in Eq. (\ref{supp_closure_bound}).

The result does not depend on the coordinates used to represent that norm. For $\|v\|_x^2=v^\dagger G_xv$ and an invertible change of variables $x'=Sx$, let us set
\begin{equation}
N_\lambda'=SN_\lambda S^{-1},\quad C'=CS^{-1},\quad
\mathcal{U}_r^{{\rm end}\prime}=S\mathcal{U}_r^{\rm end},\quad
G_x'=S^{-\dagger}G_xS^{-1}.
\label{supp_coordinate_transform}
\end{equation}
Then $q_j'=Sq_j$, $\|Sv\|_{x'}=\|v\|_x$, and $C'Sv=Cv$, so
\begin{equation}
\underline\alpha_{C'}=\underline\alpha_C,\quad \quad
\overline\beta_{C'}=\overline\beta_C,\quad \quad
\Omega_*'=\Omega_*,\quad \quad
\mathcal{R}_r^{(x')}=\mathcal{R}_r^{(x)}.
\label{supp_coordinate_invariance}
\end{equation}
Nominal generator and readout errors, and therefore their certified bounds, must be transported by the same change of coordinates. This proves coordinate invariance for a fixed physical norm. It should, however, be noted that replacing the metric itself by a different weighting changes the scientific definition of the residual and is not a mere coordinate transformation.

\subsection{Spectral-centering error}
\label{supp:shift_error}

If the analysis uses $\widehat{\lambda}=\lambda+\Delta\lambda$, its coefficients obey the exact triangular relation
\begin{equation}
q_j^{(\widehat\lambda)}
=(N_\lambda-\Delta\lambda I)^jx_\lambda
=\sum_{s=0}^j\binom{j}{s}(-\Delta\lambda)^{j-s}q_s .
\label{supp_shift_transform}
\end{equation}
This follows directly from the binomial theorem because $N_\lambda$ commutes with the identity. Equivalently, one has
\begin{equation}
e^{-\widehat\lambda t}e^{At}x_\lambda
=e^{-\Delta\lambda t}\sum_{s=0}^{\nu_\lambda-1}\frac{t^s}{s!}q_s .
\label{supp_miscentred_trajectory}
\end{equation}
For $\Delta\lambda\neq0$, the extra exponential has nonzero Taylor coefficients at all orders. A finite-window polynomial fit need not return the exact Taylor coefficients in Eq. (\ref{supp_shift_transform}), and its bias depends on the sampling times, truncation order, weighting, and reconstruction method, and must enter the simultaneous coefficient bounds. If the correctly centered sequence terminates at $r$, setting $j=r$ gives
\begin{equation}
q_r^{(\widehat\lambda)}=-r\Delta\lambda q_{r-1}+O(|\Delta\lambda|^2),
\quad \quad
p_r^{(\widehat\lambda)}=-r\Delta\lambda p_{r-1}+O(|\Delta\lambda|^2),
\label{supp_false_terminal}
\end{equation}
because the $s=r-1$ term is linear in $\Delta\lambda$ and every remaining term is at least quadratic. Thus centering uncertainty and finite-window bias must be propagated into every coefficient used by the nonzero, rank, or termination certificate. In particular, the closure bound needs simultaneous validity for both the predecessor and terminal coefficients.

\subsection{Deterministic stability near an exceptional point}
\label{supp:near_ep_extension}

We now extend the exact-EP construction to a nearby calibrated generator, without requiring a near-EP fit or numerical parameter scan. We assume that a common contour remains in the resolvent and encloses an isolated cluster of constant algebraic dimension throughout the calibrated neighborhood. For a sufficiently small neighborhood,
$T_{\boldsymbol{\theta}}
=\Pi_{\boldsymbol{\theta}}|_{\mathcal{G}_\lambda}:
\mathcal{G}_\lambda\to\mathcal{G}_{\boldsymbol{\theta}}$
is invertible. With $m=\dim\mathcal{G}_\lambda$, we define
\begin{equation}
\begin{gathered}
\widetilde{A}(\boldsymbol{\theta})
=
T_{\boldsymbol{\theta}}^{-1}
A(\boldsymbol{\theta})|_{\mathcal{G}_{\boldsymbol{\theta}}}
T_{\boldsymbol{\theta}},
\quad \quad
\lambda_{\rm c}(\boldsymbol{\theta})
=
\frac{1}{m}\operatorname{tr}\widetilde{A}(\boldsymbol{\theta}),\\
B(\boldsymbol{\theta})
=
\widetilde{A}(\boldsymbol{\theta})
-\lambda_{\rm c}(\boldsymbol{\theta})I,
\quad \quad
N=B(\boldsymbol{\theta}_{\rm EP}).
\end{gathered}
\label{supp_near_ep_centering}
\end{equation}
The same map transports preparations, and transport and centering uncertainty is included in the bounds below.

Let $N^rx_0=0$, $a_{r-1}=\|N^{r-1}x_0\|>0$, and, at the tested nearby point, we write $B=N+E$ and $x=x_0+\eta$, where $\|E\|\leq\delta_A$ and $\|\eta\|\leq\delta_x$. We define $\Omega=\|N\|>0$, $X=\|x_0\|$ and $\mathcal{R}_r(B,x)=\|B^rx\|/[\|B\|\|B^{r-1}x\|]$, so that
\begin{equation}
\begin{gathered}
D_j=\big[(\Omega+\delta_A)^j-\Omega^j\big]X
     +(\Omega+\delta_A)^j\delta_x,\\
\|B^jx-N^jx_0\|\leq D_j,\quad \quad
\mathcal{R}_r(B,x)\leq
\frac{D_r}{(\Omega-\delta_A)(a_{r-1}-D_{r-1})},
\end{gathered}
\label{supp_near_ep_analytic_bounds}
\end{equation}
provided $\delta_A<\Omega$ and $D_{r-1}<a_{r-1}$. Indeed, the noncommutative identity $B^j-N^j=\sum_{k=0}^{j-1}B^{j-1-k}EN^k$ gives $\|B^j-N^j\|\leq(\Omega+\delta_A)^j-\Omega^j$. Applying this to $x_0$, adding the preparation error, and using the reverse triangle inequality proves Eq. (\ref{supp_near_ep_analytic_bounds}). Consequently, $D_j<\|N^jx_0\|$ preserves the $j$th nonzero direction, while the second inequality bounds the reopening of an exact endpoint.

The same control preserves the number of independent Krylov directions. To show this, we define $\mathsf{Q}_r^{(0)}=(\Omega^{-j}N^jx_0)_{j=0}^{r-1}$ and $\mathsf{Q}_r=(\Omega^{-j}B^jx)_{j=0}^{r-1}$, and let $\sigma_{\min}^{\rm col}(\mathsf{Q})=\inf_{\|a\|_2=1}\|\mathsf{Q}a\|$, using the same state norm as in the perturbation bounds. Since the $j$th column of $\mathsf{Q}_r-\mathsf{Q}_r^{(0)}$ has norm at most $D_j/\Omega^j$, the Cauchy-Schwarz and reverse triangle inequalities give
\begin{equation}
\sigma_{\min}^{\rm col}(\mathsf{Q}_r)
\geq
\sigma_{\min}^{\rm col}(\mathsf{Q}_r^{(0)})
-\left[\sum_{j=0}^{r-1}\frac{D_j^2}{\Omega^{2j}}\right]^{1/2}.
\label{supp_near_ep_rank_stability}
\end{equation}
A positive right-hand side guarantees $r$ independent physical Krylov directions. It should be noted that this is a statement about the dimension reached by the preparation, and not a Jordan-depth assignment away from coalescence.

Finally, a residual bounds parameter distance only along perturbations detected by the preparation. If $B(h)=N+hV+O(h^2)$ and $x(h)=x_0+hx_1+O(h^2)$, differentiation gives
\begin{equation}
\mathcal{R}_r(h)=
\frac{|h|}{\Omega a_{r-1}}
\left\|
\sum_{k=0}^{r-1}N^{r-1-k}VN^kx_0+N^rx_1
\right\|
+O(h^2).
\label{supp_near_ep_transversality}
\end{equation}
A nonzero leading vector therefore gives local upper and lower bounds proportional to $|h|$. For several parameters $\boldsymbol{\theta}$, we define $F(\boldsymbol{\theta})=B(\boldsymbol{\theta})^r x(\boldsymbol{\theta})$ and $\mathcal{E}_r=F^{-1}(0)$. Let us assume that $B$ and $x$ are $C^2$, that $\mathcal{E}_r$ is an embedded $C^2$ submanifold admitting a tubular neighborhood $U$, and that constants $0<d_-\leq d_+<\infty$ and $g>0$ exist such that
\begin{equation}
d_-\leq\|B(\boldsymbol{\theta})\|\|B(\boldsymbol{\theta})^{r-1}x(\boldsymbol{\theta})\|\leq d_+    
\end{equation}
throughout $U$, while $\|DF(\boldsymbol{\vartheta})n\|\geq g\|n\|$ for every $\boldsymbol{\vartheta}\in\mathcal{E}_r\cap U$ and every normal vector $n$ to $\mathcal{E}_r$ at $\boldsymbol{\vartheta}$. After shrinking $U$ so that the relevant derivatives are bounded, Taylor's theorem gives constants $0<c\leq C<\infty$ such that
\begin{equation}
c\operatorname{dist}(\boldsymbol{\theta},\mathcal{E}_r)
\leq\mathcal{R}_r(\boldsymbol{\theta})
\leq C\operatorname{dist}(\boldsymbol{\theta},\mathcal{E}_r)    
\end{equation}
throughout $U$. If either uniform condition fails, no uniform first-order distance bound follows, and higher-order terms must be examined separately.

\section{Optomechanical model and finite-imprecision certification}\label{supp:optomech_realization}

This section derives the closed moment model used in the numerical illustration, specifies how signed preparations are implemented by physical states, and records the complete measurement and confidence model. The simulation is an exact-EP consistency test of the certification procedure, so model departures that are not injected into the simulated data are stated explicitly as experimental caveats.

\subsection{Coherence-sector dynamics and preparations}

At the resonant red sideband, the rotating-wave Langevin equations are \cite{Aspelmeyer_2014,Ghosh_2026}
\begin{equation}
\dot a=-\frac{\kappa}{2}a-iGb+\sqrt{\kappa}a_{\rm in},
\quad \quad
\dot b=-\frac{\gamma_{\rm m}}{2}b-iGa+\sqrt{\gamma_{\rm m}}b_{\rm in}.
\label{supp_optomech_langevin}
\end{equation}
The first moments evolve under
\begin{equation}
K=\begin{pmatrix}-\kappa/2&-iG\\-iG&-\gamma_{\rm m}/2\end{pmatrix},
\quad \quad
\lambda_\pm=-\frac{\kappa+\gamma_{\rm m}}4
\pm\sqrt{\frac{(\kappa-\gamma_{\rm m})^2}{16}-G^2}.
\label{supp_optomech_first_moment}
\end{equation}
For $\kappa>\gamma_{\rm m}$ they coalesce at $G_{\rm EP}=(\kappa-\gamma_{\rm m})/4$. Writing $K=-\bar\Gamma I/2+N_1$, with $\bar\Gamma=(\kappa+\gamma_{\rm m})/2$ and $N_1=G_{\rm EP}\bigl(\begin{smallmatrix}-1&-i\\-i&1\end{smallmatrix}\bigr)$, gives $N_1^2=0$ and $N_1\neq0$.

For the normally ordered coherence matrix $S_{ij}=\langle v_i^\dagger v_j\rangle$, $\boldsymbol{v}=(a,b)^T$, independent phase-insensitive Markovian baths give the Lyapunov dynamics
\begin{equation}
\dot S=K^\dagger S+SK+D,\quad \quad
D=\operatorname{diag}(\kappa\bar n_{\rm c},\gamma_{\rm m}\bar n_{\rm m}),
\quad \quad K^\dagger S_{\rm ss}+S_{\rm ss}K+D=0.
\label{supp_optomech_lyapunov}
\end{equation}
After subtracting the stationary solution, $X=S-S_{\rm ss}$ obeys $\dot X=-\bar\Gamma X+N_1^\dagger X+XN_1$. The shifted map $\mathcal{N}_{\rm coh}(X)=N_1^\dagger X+XN_1$ satisfies
\begin{equation}
\mathcal{N}_{\rm coh}^2(X)=2N_1^\dagger XN_1,\quad \quad
\mathcal{N}_{\rm coh}^3=0,
\label{supp_optomech_lift_nilpotency}
\end{equation}
and its square is nonzero, for example on $X=\operatorname{diag}(1,0)$. In the invariant coordinates $\boldsymbol{x}=(\delta n_{\rm c},\delta n_{\rm m},\delta Y_{\rm cm})^T$, where $\delta Y_{\rm cm}=2{\rm Im}[\delta\langle a^\dagger b\rangle]$, one gets
\begin{equation}
\dot{\boldsymbol{x}}=(-\bar\Gamma I+N_{\rm coh})\boldsymbol{x},
\quad \quad
N_{\rm coh}=G_{\rm EP}
\begin{pmatrix}-2&0&1\\0&2&-1\\-2&2&0\end{pmatrix},
\quad \quad N_{\rm coh}^3=0\neq N_{\rm coh}^2.
\label{supp_optomech_Ncoh}
\end{equation}
The lifted index three follows without diagonalizing the moment drift: each action of $\mathcal{N}_{\rm coh}$ inserts $N_1$ on the left or right, and any third action contains either $N_1^2$ or $(N_1^\dagger)^2$. The displayed nonzero square excludes index one or two. Thus an EP2 in the amplitude drift produces an available length-three chain in this bilinear sector.

The complementary coordinate $2{\rm Re}[\delta\langle a^\dagger b\rangle]$ decays independently at $-\bar\Gamma$, giving the $3\oplus1$ partition. Both the blocks share the same eigenvalue, so the three-coordinate sector is selected by its independently validated invariance, not by spectral separation. This decomposition is exact within the resonant rotating-wave model. Residual detuning couples the complementary coordinate back into the retained sector, while counterrotating terms couple number-conserving bilinears to pairing moments. An experiment must enlarge the calibrated sector or include these effects as coefficient bias as neither is present in the simulation.

We use
\begin{equation}
\|\boldsymbol{v}\|_{\rm coh}^2=\boldsymbol{v}^\dagger M_{\rm coh}\boldsymbol{v}
=2|v_{\rm c}|^2+2|v_{\rm m}|^2+|v_Y|^2,\quad \quad
M_{\rm coh}=\operatorname{diag}(2,2,1).
\label{supp_optomech_metric}
\end{equation}
For $X(\boldsymbol{v})=\bigl(\begin{smallmatrix}v_{\rm c}&iv_Y/2\\-iv_Y/2&v_{\rm m}\end{smallmatrix}\bigr)$, this is $\|\boldsymbol{v}\|_{\rm coh}=\sqrt2\|X(\boldsymbol{v})\|_{\rm F}$. Its induced rate is
\begin{equation}
\Omega_*=\|N_{\rm coh}\|_{\rm coh\to coh}
=\frac{\kappa-\gamma_{\rm m}}{\sqrt2}=2\sqrt2G_{\rm EP}.
\label{supp_optomech_rate}
\end{equation}
With $d=(\kappa-\gamma_{\rm m})/2$, the cavity-noise preparation generates
\begin{equation}
\boldsymbol{q}_0^{(3)}=n_{\rm p}(1,0,0)^T,\quad
\boldsymbol{q}_1^{(3)}=-dn_{\rm p}(1,0,1)^T,\quad
\boldsymbol{q}_2^{(3)}=\frac{d^2n_{\rm p}}2(1,1,2)^T,\quad
\boldsymbol{q}_3^{(3)}=0,
\label{supp_optomech_depth3_chain}
\end{equation}
whereas the differential preparation gives
\begin{equation}
\boldsymbol{q}_0^{(2)}=n_{\rm p}(1,-1,0)^T,\quad
\boldsymbol{q}_1^{(2)}=-dn_{\rm p}(1,1,2)^T,\quad
\boldsymbol{q}_2^{(2)}=0.
\label{supp_optomech_depth2_chain}
\end{equation}
The first chain has rank three and the second rank two. The signed differential preparation is implemented by subtracting two physical states in the manner
\begin{equation}
S_A(0)=S_{\rm ss}+n_{\rm p}\operatorname{diag}(1,0),\quad \quad
S_B(0)=S_{\rm ss}+n_{\rm p}\operatorname{diag}(0,1),
\label{supp_optomech_depth2_physical_pair}
\end{equation}
rather than by assigning a negative occupation. Both normal covariance matrices are positive semidefinite. Their anomalous covariances vanish and remain zero under the passive red-sideband dynamics, so positive semidefiniteness of $S$ is sufficient for physicality within this gauge-invariant Gaussian family. Their common diffusion cancels from
\begin{equation}
\frac{d}{dt}(S_A-S_B)=K^\dagger(S_A-S_B)+(S_A-S_B)K,
\label{supp_optomech_differential_homogeneous}
\end{equation}
while independent measurement covariances add. The code checks positivity of both total trajectories at every acquisition time. Nilpotency gives
\begin{equation}
\boldsymbol{x}(t)=e^{-\bar\Gamma t}
\left(\boldsymbol{q}_0+t\boldsymbol{q}_1+\frac{t^2}{2}\boldsymbol{q}_2\right).
\label{supp_optomech_polynomial_response}
\end{equation}
The chains in Eqs. (\ref{supp_optomech_depth3_chain}) and (\ref{supp_optomech_depth2_chain}) are obtained by direct multiplication with Eq. (\ref{supp_optomech_Ncoh}). Their nonzero prefixes are independent, so they establish activated depths three and two in the exact model. The solid curves fitted to the simulated records use the same polynomial form, but the terminal order is included as a free coefficient and is never imposed from these identities.

\subsection{Closed-form optomechanical stability near the exceptional point}
\label{supp:optomech_near_ep}

The optomechanical coherence sector permits a sharper result than the general perturbation bound within the resonant red-sideband rotating-wave family. Taking $G>0$ by phase convention, let $d=(\kappa-\gamma_{\rm m})/2>0$ and $s=d^2-4G^2$. The shifted coherence drift then obeys
\begin{equation}
B_{\rm coh}=
\begin{pmatrix}
-d&0&G\\
0&d&-G\\
-2G&2G&0
\end{pmatrix},
\quad \quad
B_{\rm coh}^3=sB_{\rm coh}, \quad \quad
\|B_{\rm coh}\|_{\rm coh\to coh}
=\sqrt{d^2+4G^2}.
\label{supp_near_ep_optomech_algebra}
\end{equation}
The cubic identity gives the exact propagator
\begin{equation}
\begin{aligned}
e^{B_{\rm coh}t}
&=I+S_s(t)B_{\rm coh}+F_s(t)B_{\rm coh}^2,\\
S_s(t)&=\sum_{n\geq0}\frac{s^nt^{2n+1}}{(2n+1)!},
\quad \quad
F_s(t)=\sum_{n\geq0}\frac{s^nt^{2n+2}}{(2n+2)!}.
\end{aligned}
\label{supp_near_ep_propagator}
\end{equation}
At the EP, $s=0$, so $S_0(t)=t$ and $F_0(t)=t^2/2$, recovering the terminating Jordan polynomial.

For the differential preparation, let $q_0^{(2)}=n_{\rm p}(1,-1,0)^T$ and $q_j^{(2)}=B_{\rm coh}^jq_0^{(2)}$. Direct multiplication gives $q_1^{(2)}=-n_{\rm p}(d,d,4G)^T$ and $q_2^{(2)}=s q_0^{(2)}$. The first two vectors remain independent, while the exact endpoint becomes a recurrence within their two-dimensional cyclic space. Using the amplitude eigenvalues $\lambda_\pm=-\bar\Gamma/2\pm\sqrt{s}/2$ gives
\begin{equation}
\begin{gathered}
\mathcal{R}_2^{(2)}
=\frac{|d^2-4G^2|}{d^2+4G^2}
=\frac{|\lambda_+-\lambda_-|^2}{d^2+4G^2}
=\frac{|1-u^2|}{1+u^2}
=\tanh|\log u|,
\quad \quad
u=\frac{4G}{\kappa-\gamma_{\rm m}},\\
0\leq b<1,\quad \quad
\mathcal{R}_2^{(2)}\leq b
\quad\Longleftrightarrow\quad
\sqrt{\frac{1-b}{1+b}}
\leq\frac{4G}{\kappa-\gamma_{\rm m}}
\leq\sqrt{\frac{1+b}{1-b}}.
\end{gathered}
\label{supp_near_ep_optomech_residual}
\end{equation}
Thus the endpoint reopening equals the squared spectral splitting normalized by $\|B_{\rm coh}\|_{\rm coh\to coh}^2$, and a prescribed residual tolerance gives an explicit interval for the coupling. For the cavity-noise preparation, $\det(q_0^{(3)},q_1^{(3)},q_2^{(3)})=4n_{\rm p}^3G^3$, so its cyclic dimension remains three for every $G\neq0$. The same relations apply to the Bose-Einstein-condensate specialization discussed subsequently, after the identification $\gamma_{\rm m}\mapsto\gamma_{\rm B}$, subject to its stated two-mode and rotating-wave assumptions.

\subsection{Calibrated multibranch readout and noise model}

Independent repetitions reconstruct the two occupations and the coherence. Opposite phases of a calibrated $50{:}50$ mixing pulse yield
\begin{equation}
n_-=(n_{\rm c}+n_{\rm m}+Y_{\rm cm})/2,
\quad \quad n_+=(n_{\rm c}+n_{\rm m}-Y_{\rm cm})/2,
\quad \quad Y_{\rm cm}=n_--n_+.
\label{supp_optomech_correlation_branches}
\end{equation}
The code uses these ideal calibrated mean relations rather than the internal mixing-pulse dynamics; loss and added noise enter through effective gains and floors. Combining branches gives the nominal map $\widetilde{C}=I_3$, which denotes reconstruction in common coordinates, not simultaneous measurement. We choose
\begin{equation}
\mathcal{U}_2^{\rm end}=\mathcal{U}_3^{\rm end}=\mathcal{M}_{\rm coh}=\mathbb{C}^3,
\label{supp_optomech_terminal_sector}
\end{equation}
which contains both endpoint pairs. More explicitly, the effective mean model is
\begin{equation}
\widehat{\boldsymbol{z}}=C\boldsymbol{x}+\boldsymbol\xi,\quad \quad
C=I_3-\Delta C,\quad \quad \mathbb{E}\boldsymbol\xi=0,\quad \quad
\|\Delta C\|_{\rm coh\to coh}\leq\delta_{\rm ro}<1.
\label{supp_optomech_readout_decomposition}
\end{equation}
Thus
\begin{equation}
\underline\alpha_C^{\rm L}=1-\delta_{\rm ro},
\quad \quad \overline\beta_C^{\rm U}=1+\delta_{\rm ro},
\quad \quad \mathcal{C}_{\rm ro}=\frac{1+\delta_{\rm ro}}{1-\delta_{\rm ro}}.
\label{supp_optomech_readout_factor}
\end{equation}
The synthetic means use $\Delta C=0$, so the declared margin is used only in the certificate and is not counted again as statistical noise.

Every branch is acquired on fresh repetitions at each evolution time. The two mixing phases are also independent records, i.e., only after gain calibration are they combined into $Y_{\rm cm}$. The simulation evaluates all branch variances from their total physical occupations and checks that the total occupations entering the power model are nonnegative. Consequently, neither a signed excess nor the signed differential trajectory is ever inserted into a detector variance.

For branch $\mu$, time $t_\ell$, and repetition $s$, the normalized heterodyne outcome obeys
\begin{equation}
\beta_{\mu\ell s}\sim\mathcal{CN}(0,\mu_{\mu\ell}),\quad \quad
\mu_{\mu\ell}=\nu_{{\rm det},\mu}+g_\mu n_\mu(t_\ell).
\label{supp_optomech_heterodyne_mean}
\end{equation}
Here $n_\mu$ is always the total occupation of a physical mode, never a signed excess. In the convention where vacuum heterodyne power is $1/2$, the corresponding phase-insensitive channel has $\tau_\mu=2g_\mu$ and added output occupation $n_{0,\mu}=2\nu_{{\rm det},\mu}-1$. Complete positivity requires $n_{0,\mu}\geq\max(0,\tau_\mu-1)$, equivalently
\begin{equation}
\nu_{{\rm det},\mu}\geq\max(1/2,g_\mu).
\label{supp_optomech_channel_physicality}
\end{equation}
Because $|\beta|^2$ is exponential \cite{GoodmanSupp}, the averaged power $\overline Q_{\mu\ell}$ obeys
\begin{equation}
\mathbb{E}\overline Q_{\mu\ell}=\mu_{\mu\ell},\quad \quad
\operatorname{Var}\overline Q_{\mu\ell}=\frac{\mu_{\mu\ell}^2}{M_\mu}.
\label{supp_optomech_averaged_power_moments}
\end{equation}
For the cavity-noise record, $\widehat{\delta n}_{\mu\ell}=(\overline Q_{\mu\ell}-\overline Q_{\mu,{\rm ss}})/g_\mu$, and one stationary reference is reused at every time in each branch, giving
\begin{equation}
\operatorname{Cov}(\widehat{\delta n}_{\mu\ell},\widehat{\delta n}_{\mu m})
=\delta_{\ell m}\!\left[\frac{\mu_{\mu\ell}^2}{M_\mu g_\mu^2}+\sigma_{{\rm tech},\mu}^2\right]
+\frac{\mu_{\mu,{\rm ss}}^2}{M_{\mu,{\rm ss}}g_\mu^2}.
\label{supp_optomech_shared_reference_covariance}
\end{equation}
The last term is the temporal correlation from the shared finite reference. References belonging to different branches are independent. For the differential record, $\widehat{\Delta n}_{\mu\ell}=(\overline Q^A_{\mu\ell}-\overline Q^B_{\mu\ell})/g_\mu$, and independently measured $A$ and $B$ trajectories give
\begin{equation}
\operatorname{Cov}(\widehat{\Delta n}_{\mu\ell},\widehat{\Delta n}_{\mu m})
=\delta_{\ell m}\!\left[\frac{(\mu^A_{\mu\ell})^2+(\mu^B_{\mu\ell})^2}{M_\mu g_\mu^2}+\sigma_{{\rm tech},\mu}^2\right].
\label{supp_optomech_differential_covariance}
\end{equation}
The independent $n_-$ and $n_+$ covariances add when reconstructing $Y_{\rm cm}$, after which the declared $\sigma_{{\rm tech},Y}^2$ is added once. It is not assigned separately to the constituent records.

These two covariance formulas encode different correlations. Reusing one stationary reference produces the rank-one temporal term in Eq. (\ref{supp_optomech_shared_reference_covariance}); the independently prepared differential trajectories have no such shared reference and therefore give the diagonal covariance in Eq. (\ref{supp_optomech_differential_covariance}). Statistical fluctuations of independent physical records add under subtraction. The technical rms parameter denotes one final reconstructed-channel contribution, which is why adding it to each constituent record would double count it.

The stacked covariances $\Sigma_z^{(s)}$ are checked to be positive definite, and the effective records are drawn as
\begin{equation}
\widehat{\boldsymbol{z}}^{(s)}\sim\mathcal{N}(\boldsymbol{m}_z^{(s)},\Sigma_z^{(s)}),
\quad \quad s\in\{2,3\}.
\label{supp_optomech_gaussian_draw}
\end{equation}
Here $\boldsymbol{m}_z^{(s)}=\operatorname{col}_\ell\boldsymbol{x}^{(s)}(t_\ell)$, with time-major stacking. Eq. (\ref{supp_optomech_gaussian_draw}) is exact for the declared effective model. As a large-$M_\mu$ approximation to averaged exponential powers, its smallest averaging number gives standardized skewness $2/\sqrt{2\times10^6}=1.42\times10^{-3}$. A real experiment would independently calibrate the covariance or use a conservative upper bound. The synthetic draw uses exact spectral centering and injects neither a generator perturbation nor deterministic coefficient bias.

Demodulation uses $D_{\bar\Gamma}=\operatorname{diag}(e^{\bar\Gamma t_\ell})\otimes I_3$ and transforms both record and covariance:
\begin{equation}
\widehat{\boldsymbol{Y}}=D_{\bar\Gamma}\widehat{\boldsymbol{z}},\quad \quad
\Sigma_Y=D_{\bar\Gamma}\Sigma_zD_{\bar\Gamma}^T.
\label{supp_optomech_demodulation}
\end{equation}
With $\tau=t/t_{\rm sc}$ and $\boldsymbol{c}_j=t_{\rm sc}^j\boldsymbol{p}_j$, generalized least squares (GLS) \cite{SeberLeeSupp} for $\widehat{\boldsymbol{Y}}_\ell=\sum_{j=0}^m\tau_\ell^j\boldsymbol{c}_j/j!+\boldsymbol\zeta_\ell$, with $H_{\ell j}=(\tau_\ell^j/j!)I_3$, gives
\begin{equation}
\widehat{\boldsymbol{c}}=(H^T\Sigma_Y^{-1}H)^{-1}H^T\Sigma_Y^{-1}\widehat{\boldsymbol{Y}},
\quad \quad \Gamma_c=(H^T\Sigma_Y^{-1}H)^{-1}.
\label{supp_optomech_GLS}
\end{equation}
The design contains the factorial convention of the Krylov expansion, and the full $\Sigma_Y$ retains the shared-reference correlations. Because the design is fixed and has full column rank for the stated distinct times, the GLS estimator is unbiased in the declared model and has covariance $\Gamma_c$.

We fit through cubic order for the depth-three record and quadratic order for the depth-two record, so each terminal coefficient is estimated rather than set to zero. Under the effective Gaussian model, for $D=3(m+1)$,
\begin{equation}
(\widehat{\boldsymbol{c}}-\boldsymbol{c})^T\Gamma_c^{-1}
(\widehat{\boldsymbol{c}}-\boldsymbol{c})\leq\chi^2_{D,0.95}
\label{supp_optomech_joint_ellipsoid}
\end{equation}
has probability $0.95$. Projection of this single ellipsoid gives
\begin{equation}
\varepsilon_{c,j}^2=\chi^2_{D,0.95}\lambda_{\max}\!\left(M_{\rm coh}^{1/2}E_j\Gamma_cE_j^TM_{\rm coh}^{1/2}\right).
\label{supp_optomech_epsilon}
\end{equation}
To see the projection, we maximize $\|E_j(\widehat{\boldsymbol{c}}-\boldsymbol{c})\|_{\rm coh}^2$ over the joint ellipsoid. Whitening by $\Gamma_c^{1/2}$ reduces the maximum to the largest eigenvalue in Eq. (\ref{supp_optomech_epsilon}). The resulting block radii are therefore simultaneous consequences of one coefficient event, not separate marginal intervals.

All coefficients for one preparation are controlled on the same event and the two preparations use separate $95\%$ events, with no $95\%$ claim for their intersection. Nonzero coefficients are tested using Eq. (\ref{supp_nonzero_certificate}). With $\Omega_{\rm sc}=1/t_{\rm sc}$, the rank condition is
\begin{equation}
\sigma_{\min}^{\rm col}\!\left[
M_{\rm coh}^{1/2}(\widehat{\boldsymbol{c}}_0,\ldots,\widehat{\boldsymbol{c}}_{r-1})
\right]>
\left(\sum_{j=0}^{r-1}\varepsilon_{c,j}^2\right)^{1/2}.
\label{supp_optomech_scaled_rank_test}
\end{equation}

\subsection{Numerical illustration and certificates}

We take $(\omega_{\rm m},\kappa,\gamma_{\rm m})/2\pi=(1~{\rm MHz},0.2~{\rm MHz},50~{\rm Hz})$, so $G_{\rm EP}/2\pi=49.9875~{\rm kHz}$, together with $T_{\rm m}=4~{\rm K}$, $n_{\rm p}=10$, $t_{\rm sc}=1~\mu{\rm s}$, and $41$ equally spaced times on $0\leq t\leq5~\mu{\rm s}$. Rates in the matrices are angular frequencies. We set the optical bath occupation to zero and the mechanical Bose occupation is $8.334597649\times10^4$. Solving Eq. (\ref{supp_optomech_lyapunov}) gives
\begin{equation}
S_{\rm ss}\simeq
\begin{pmatrix}20.8104654274&41.6313386895i\\-41.6313386895i&104.114784556\end{pmatrix},
\quad \quad Y_{{\rm cm},{\rm ss}}=83.2626773791,\quad \quad
\lambda_{\min}(S_{\rm ss})=3.572296775>0.
\label{supp_optomech_stationary_matrix_numerical}
\end{equation}
These are coupled-system stationary moments, not bath occupations, and transfer from the mechanical bath produces nonzero cavity occupation and coherence despite an optical-vacuum input.

The power gains, vacuum floors, averaging numbers, and final-channel technical rms values are
\begin{equation}
\begin{gathered}
(g_{\rm c},g_{\rm m},g_{\rm h})=(0.45,0.25,0.35),
\quad \nu_{{\rm det},\mu}=0.50,\\
(M_{\rm c},M_{\rm m},M_{\rm h})=(2,4,4)\times10^6,
\quad(\sigma_{{\rm tech},{\rm c}},\sigma_{{\rm tech},{\rm m}},\sigma_{{\rm tech},Y})=(3,5,7)\times10^{-4}.
\end{gathered}
\label{supp_optomech_numerical_readout}
\end{equation}
The corresponding efficiencies are $2g_\mu=(0.90,0.50,0.70)$ and satisfy Eq. (\ref{supp_optomech_channel_physicality}). Each cavity-noise reference uses $2M_\mu$ averages; branches and differential trajectories are mutually independent. The displayed synthetic realization uses the seed $20260722$.

The symbol $\mathrm{h}$ denotes either correlation-sensitive mixing branch; both use the stated gain, floor, and averaging number. The vacuum floors and gains represent physical attenuating channels under the normalization above. The stated reference size, independence structure, technical floors, time grid, fit orders, covariance construction, and seed define this illustrative pair of records.

Before noise is drawn, the exact scaled checks are
\begin{equation}
\|\boldsymbol{c}_2^{(3)}\|_{\rm coh}=5.580300163,\quad \quad
\|\boldsymbol{c}_1^{(2)}\|_{\rm coh}=17.76708887,\quad \quad
\boldsymbol{c}_3^{(3)}=\boldsymbol{c}_2^{(2)}=0.
\label{supp_optomech_exact_scaled_checks}
\end{equation}
The synthetic generator and mean map equal their nominal values. For this illustration, certification uses predeclared deterministic $2\%$ margins on the shifted-generator norm and readout map, fixed independently of the simulated record as
\begin{equation}
\Omega_{K,{\rm sc}}^{\rm L}=0.98t_{\rm sc}\Omega_*=0.8705873546,\quad \quad
\mathcal{C}_{\rm ro}=\frac{1.02}{0.98}=1.040816327.
\label{supp_optomech_realized_calibration_factors}
\end{equation}
The generator margin bounds the norm conditional on the independently calibrated nilpotent sector; it does not assert that an arbitrary $2\%$ matrix perturbation preserves $\nu_\lambda=3$. Resonant red-sideband operation is likewise assumed independently; residual detuning must be included in a four-coordinate model or as bounded bias. No such bias is added to the draw.

For scaled coefficients, Eq. (\ref{supp_closure_bound}) becomes
\begin{equation}
B_r=\mathcal{C}_{\rm ro}
\frac{\|\widehat{\boldsymbol{c}}_r\|_{\rm coh}+\varepsilon_{c,r}}
{\Omega_{K,{\rm sc}}^{\rm L}(\|\widehat{\boldsymbol{c}}_{r-1}\|_{\rm coh}-\varepsilon_{c,r-1})},
\quad \quad \mathcal{R}_r\leq B_r.
\label{supp_optomech_specialized_bounds}
\end{equation}
The rescaling $\boldsymbol{c}_j=t_{\rm sc}^j\boldsymbol{p}_j$ together with $\Omega_{K,{\rm sc}}^{\rm L}=t_{\rm sc}\Omega_K^{\rm L}$ changes conditioning but not the physical residual. The realized norms and projected radii are
\begin{equation}
\begin{array}{c|cc}
&\|\widehat{\boldsymbol{c}}_j\|_{\rm coh}&\varepsilon_{c,j}\\ \hline
(3),\,j=2&5.212816954&1.153546707\\
(3),\,j=3&0.393908922&0.670519769\\
(2),\,j=1&17.66421801&0.552583801\\
(2),\,j=2&0.088929204&0.358618701
\end{array}
\label{supp_optomech_realized_norms}
\end{equation}
and hence predecessor/terminal ratios $(4.51895,0.58747)$ and $(31.96659,0.24798)$. Thus the predecessors, but not the terminal coefficients, are certified as nonzero. As noted earlier, a ratio below unity is not a zero inference. The rank margins are $1.98887>1.46076$ and $17.66416>0.624801$, certifying $\operatorname{rank}P_3=3$ and $\operatorname{rank}P_2=2$ on their respective events. The simulation-only Mahalanobis checks are $16.3682<\chi^2_{12,0.95}=21.0261$ and $5.6517<\chi^2_{9,0.95}=16.9190$, so the displayed draws lie in their separate events. The latter diagnostic is unavailable for experimental data because the true coefficients are unknown.

For the cavity-noise preparation, the certified quadratic coefficient and the independently validated sector index $\nu_\lambda=3$ together force $d_{\rm act}=3$; the rank margin further shows that all three reached directions are independently resolved. For the differential preparation, the linear coefficient and rank margin establish two reached and independently resolved directions. The terminal estimates are retained only to form conservative upper bounds.

Finally,
\begin{equation}
\mathcal{R}_3\leq B_3=0.31349480,
\quad \quad
\mathcal{R}_2\leq B_2=0.03126870.
\label{supp_optomech_numerical_bounds}
\end{equation}
Both simulated terminal vectors vanish analytically: $N_{\rm coh}^3=0$ fixes $\boldsymbol{q}_3^{(3)}=0$, while $N_{\rm coh}^2\boldsymbol{x}_0^{(2)}=0$ fixes $\boldsymbol{q}_2^{(2)}=0$. The positive $B_r$ are conservative finite-error bounds obtained without substituting these above identities. For the stated $\tau_{\rm term}=0.10$, the general inequality certifies the tolerance only for depth two, but its failure for depth three does not imply $\mathcal{R}_3>\tau_{\rm term}$. In the full exact-EP model, the identities fix both endpoints, while the noisy records certify the nonzero predecessors and ranks. 

All probability statements in this subsection are conditional on the declared Gaussian record model and the relevant preparation-specific event. The two margins are fixed deterministic inputs, not random variables or estimates obtained from the simulated record. The reported $B_r$ therefore combine finite coefficient precision with calibrated gain and rate uncertainty, but do not assign statistical distributions to the calibration margins.

The reported numbers describe one fixed illustrative seed and are not used to claim typical performance over repeated noise realizations. Any optimization using an experimental record would require independent data or inclusion of all alternatives in the joint uncertainty construction.

\subsection{Alternate platform: BEC-in-cavity specialization}
\label{supp:bec_cavity_specialization}

An alternate physical realization can be obtained from Bose-Einstein condensates (BECs) in optical cavities \cite{KesslerSupp,WolkeSupp,KlinderSupp}. For the $^{87}{\rm Rb}$ cavity platform of Ke{\ss}ler \textit{et al.}, we take $N=10^5$, $\Delta_0/2\pi=-0.50~{\rm Hz}$, $\omega_{\rm rec}/2\pi=3.55~{\rm kHz}$, and field-amplitude decay $\kappa_{\rm H}/2\pi=4.50~{\rm kHz}$ \cite{KesslerSupp}. Their convention $\dot a=-\kappa_{\rm H}a+\cdots$ gives $\kappa/2\pi=9.00~{\rm kHz}$ here. In the undepleted, collisionless two-mode limit, the symmetric density mode has $\Omega_{\rm B}=4\omega_{\rm rec}$ and $g_0=\Delta_0\sqrt{N/8}$. After linearization about the driven steady state and tuning to the red sideband, one has the interaction Hamiltonian
\begin{equation}
H_{\rm int}=\hbar G(\delta a^\dagger\delta b+\delta a\delta b^\dagger),
\quad \quad G=|g_0|\sqrt{\bar n},
\label{supp_bec_red_sideband}
\end{equation} for the photon- and phonon-fluctuation operators $\delta a$ and $\delta b$. Choosing the illustrative value $\gamma_{\rm B}/2\pi=0.20~{\rm kHz}$ yields $\Omega_{\rm B}/2\pi=14.20~{\rm kHz}$, $G_{\rm EP}/2\pi=2.20~{\rm kHz}$, $|g_0|/2\pi=55.90~{\rm Hz}$, and $\bar n_{\rm EP}=1.5488\times10^3$. The static displacement $|\beta|\simeq |g_0|\bar n_{\rm EP}/\Omega_{\rm B}=6.10$ gives $|\beta|^2/N\simeq3.7\times10^{-4}$, consistent with negligible depletion. The coherence drift is Eq. (\ref{supp_optomech_Ncoh}) with $\gamma_{\rm m}\mapsto\gamma_{\rm B}$. Illustrative finite-imprecision records for the two preparations are shown in Fig. \ref{fig:bec_krylov}.

\begin{figure}[t]
\centering
\begin{minipage}[t]{0.488\columnwidth}\centering
\includegraphics[width=\linewidth]{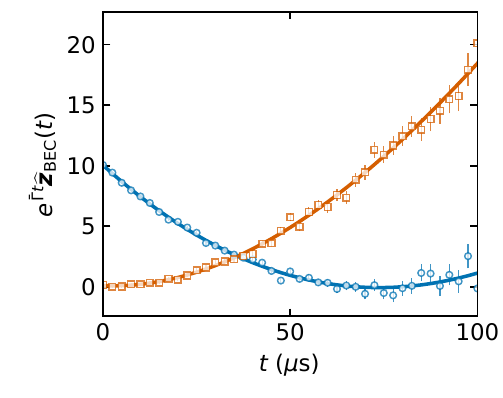}\par{\sffamily\small (a)}
\end{minipage}\hfill
\begin{minipage}[t]{0.488\columnwidth}\centering
\includegraphics[width=\linewidth]{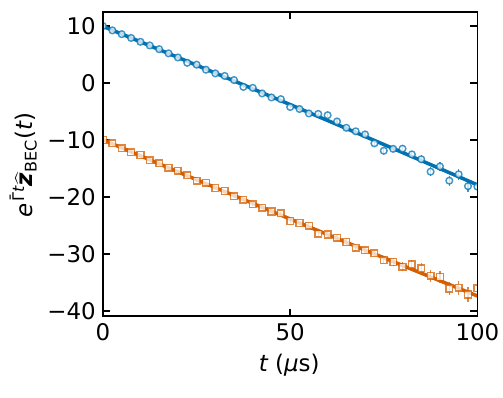}\par{\sffamily\small (b)}
\end{minipage}
\caption{\justifying{Simulated demodulated cavity and BEC-sidemode records for the (a) cavity-noise and (b) differential preparations. Blue circles denote the cavity coordinate and orange squares the BEC sidemode coordinate. In (b), the code samples the final signed coordinates directly and assigns the declared rms imprecisions to those records; it does not simulate the constituent physical measurements. Error bars are marginal standard deviations.}}
\label{fig:bec_krylov}
\end{figure}

The experimentally motivated nominal readout is
\begin{equation}
\boldsymbol{z}_{\rm BEC}=C_{\rm BEC}\boldsymbol{x},\quad \quad
C_{\rm BEC}=\begin{pmatrix}1&0&0\\0&1&0\end{pmatrix},
\label{supp_bec_readout_map}
\end{equation}
corresponding to cavity transmission and time-of-flight momentum imaging \cite{WolkeSupp,KlinderSupp}. It cannot have projected rank three, although a nonzero quadratic output still certifies depth-three activation. For the differential preparation, the projected vectors $n_{\rm p}(1,-1)^T$ and $-dn_{\rm p}(1,1)^T$ are independent. Because $C_{\rm BEC}$ is not injective on the full endpoint sector, this two-coordinate illustration does not claim the state-level closure bound of Eq. (\ref{supp_closure_bound}).

The simulation uses $n_{\rm p}=10$, $t_{\rm sc}=20~\mu{\rm s}$, $41$ points on $0$--$100~\mu{\rm s}$, seed $20260728$, and independently added final-record rms imprecisions $(0.060,0.080)$. These illustrative values are not inferred from the cited experiments and for the differential case the code samples the final signed record directly without constructing the constituent trajectories. The same cubic and quadratic polynomial fits are used, now with two output coordinates and independent Gaussian imprecision at each time. The readout example is intended to establish transferable activation and rank statements, not a complete detector model for either cited experiment.

Separate joint-$95\%$ ellipsoids give
\begin{equation}
(r_2^{(3)},r_3^{(3)})=(2.63722,0.413718),\quad \quad
(r_1^{(2)},r_2^{(2)})=(24.5630,0.424716),
\label{supp_bec_coefficient_ratios}
\end{equation}
where $r_j^{(s)}=\|\widehat{\boldsymbol{c}}_j^{(s)}\|_2/\varepsilon_{c,j}^{(s)}$. Thus the depth-three predecessor is nonzero and, with calibrated $\nu_\lambda=3$, certifies $d_{\rm act}=3$; neither terminal coefficient is certified as nonzero. For the differential record,
\begin{equation}
\sigma_{\min}^{\rm col}\!\left[
(\widehat{\boldsymbol{c}}_0^{(2)},\widehat{\boldsymbol{c}}_1^{(2)})
\right]
=7.83253>0.364921
\end{equation}
certifies $\operatorname{rank}P_2=2$. The simulation-only joint-event checks are $6.286217<15.507313$ and $4.770922<12.591587$. Residual detuning, collisions, counterrotating terms, and additional momentum modes must be included in an enlarged calibrated sector or bounded as coefficient bias in an experiment.

We conclude by emphasizing that the BEC parameters rely on the undepleted, collisionless two-mode and red-sideband approximations. The quoted depletion check supports, but does not prove, their accuracy for a specific apparatus. A quantitative experimental certificate must independently validate the selected sector and propagate uncertainty in $\gamma_{\rm B}$, the shifted detuning, mode truncation, and readout calibration.

\begingroup
\renewcommand{\refname}{Supplementary References}

\endgroup


\begin{thebibliography}{99}

\bibitem{Rotter_2009}
I. Rotter, {\it A non-Hermitian Hamilton operator and the physics of open quantum systems}, J. Phys. A: Math. Theor. \textbf{42}, 153001 (2009).

\bibitem{Berry_2004}
M. V. Berry, {\it Physics of nonhermitian degeneracies}, Czech. J. Phys. \textbf{54}, 1039 (2004).

\bibitem{Heiss_2012}
W. D. Heiss, {\it The physics of exceptional points}, J. Phys. A: Math. Theor. \textbf{45}, 444016 (2012).

\bibitem{Bagarello_2015}
F. Bagarello, J.-P. Gazeau, F. H. Szafraniec, and M. Znojil (Eds.), {\it Non-selfadjoint operators in
quantum physics: Mathematical aspects} (Wiley, Hoboken, 2015).

\bibitem{Longhi_2017}
S. Longhi, {\it Parity-time symmetry meets photonics: A new twist in non-Hermitian optics}, EPL \textbf{120}, 64001 (2017).

\bibitem{Bergholtz_2017}
E. J. Bergholtz, J. C. Budich, and F. K. Kunst, {\it Exceptional topology of non-Hermitian systems}, Rev. Mod. Phys. \textbf{93}, 015005 (2021).

\bibitem{Dembowski_2001}
C. Dembowski, H.-D. Gr\"af, H. L. Harney, A. Heine, W. D. Heiss, H. Rehfeld, and A. Richter, {\it Experimental observation of the topological structure of exceptional points}, Phys. Rev. Lett. \textbf{86}, 787 (2001).

\bibitem{Lee_2009}
S.-B. Lee, J. Yang, S. Moon, S.-Y. Lee, J.-B. Shim, S. W. Kim, J.-H. Lee, and K. An,
{\it Observation of an exceptional point in a chaotic optical microcavity}, Phys. Rev. Lett. \textbf{103}, 134101 (2009).

\bibitem{Kim_2016}
K.-H. Kim, M.-S. Hwang, H.-R. Kim, J.-H. Choi, Y.-S. No, and H.-G. Park, {\it Direct observation of exceptional points in coupled photonic-crystal lasers with asymmetric optical gains}, Nat. Commun. \textbf{7}, 13893 (2016). 

\bibitem{Bergman_2021}
A. Bergman, R. Duggan, K. Sharma, M. Tur, A. Zadok, and A. Al\`u, {\it Observation of anti-parity-time-symmetry, phase transitions and exceptional points in an optical fibre}, Nat. Commun. \textbf{12}, 486 (2021). 

\bibitem{Xiao_2021}
L. Xiao, T. Deng, K. Wang, Z. Wang, W. Yi, and P. Xue, {\it Observation of non-Bloch parity-time symmetry and exceptional points}, Phys. Rev. Lett. \textbf{126}, 230402 (2021).

\bibitem{Soleymani_2022}
S. Soleymani, Q. Zhong, M. Mokim, S. Rotter, R. El-Ganainy, and \c{S}. K. \"Ozdemir, {\it Chiral and degenerate perfect absorption on exceptional surfaces}, Nat. Commun. \textbf{13}, 599 (2022).

\bibitem{Cartarius_2011}
H. Cartarius and N. Moiseyev, {\it Fingerprints of exceptional points in the survival probability of resonances in atomic spectra}, Phys. Rev. A \textbf{84}, 013419 (2011).

\bibitem{Minganti_2019}
F. Minganti, A. Miranowicz, R. W. Chhajlany, and F. Nori, {\it Quantum exceptional points of non-Hermitian Hamiltonians and Liouvillians: The effects of quantum jumps}, Phys. Rev. A \textbf{100}, 062131 (2019).

\bibitem{Naghiloo_2019}
M. Naghiloo, M. Abbasi, Y. N. Joglekar, and K. W. Murch, {\it Quantum state tomography across the exceptional point in a single dissipative qubit}, Nat. Phys. \textbf{15}, 1232 (2019).

\bibitem{ChenJordan_2020}
H.-Z. Chen, T. Liu, H.-Y. Luan, R.-J. Liu, X.-Y. Wang, X.-F. Zhu, Y.-B. Li, Z.-M. Gu, S.-J. Liang, H. Gao, L. Lu, L. Ge, S. Zhang, J. Zhu, and R.-M. Ma, {\it Revealing the missing dimension at an exceptional point}, Nat. Phys. \textbf{16}, 571 (2020).

\bibitem{ChenTomography_2025}
Y.-Y. Chen, K. Li, L. Zhang, Y.-K. Wu, J.-Y. Ma,
H.-X. Yang, C. Zhang, B.-X. Qi, Z.-C. Zhou, P.-Y. Hou, Y. Xu, and L.-M. Duan, {\it Quantum tomography of a third-order exceptional point in a dissipative trapped ion}, Nat. Commun. \textbf{16}, 7478 (2025).

\bibitem{Hashemi_2022}
A. Hashemi, K. Busch, D. N. Christodoulides,
\c{S}. K. \"Ozdemir, and R. El-Ganainy, {\it Linear response theory of open systems with exceptional points}, Nat. Commun. \textbf{13}, 3281 (2022).

\bibitem{Wiersig_2022}
J. Wiersig, {\it Response strengths of open systems at exceptional points}, Phys. Rev. Research \textbf{4}, 023121 (2022).

\bibitem{Horn_2013}
R. A. Horn and C. R. Johnson, {\it Matrix analysis}, 2nd ed.
(Cambridge University Press, Cambridge, 2013).

\bibitem{Prosen_2008}
T. Prosen, {\it Third quantization: A general method to solve master equations for quadratic open Fermi systems}, New J. Phys. \textbf{10}, 043026 (2008).

\bibitem{Barthel_2022}
T. Barthel and Y. Zhang, {\it Solving quasi-free and quadratic Lindblad master equations for open fermionic and bosonic systems}, J. Stat. Mech. \textbf{2022}, 113101 (2022).

\bibitem{Buca_2012}
B. Bu\v{c}a and T. Prosen, {\it A note on symmetry reductions of the Lindblad equation: Transport in constrained open spin chains}, New J. Phys. \textbf{14}, 073007 (2012).

\bibitem{Zhang_2025}
D.-J. Zhang and D. M. Tong, {\it Krylov shadow tomography: Efficient estimation of quantum Fisher information}, Phys. Rev. Lett. \textbf{134}, 110802 (2025).

\bibitem{Bourgeois_2026}
J. Bourgeois, G. Blasi, and G. Haack, {\it Transport approach to quantum state tomography}, Phys. Rev. Lett. \textbf{136}, 010802 (2026).

\bibitem{Peruzzo_2026}
M. Peruzzo, T. Grigoletto, and F. Ticozzi, {\it Reconstructing quantum states and expectations via dynamical tomography}, Phys. Rev. A \textbf{113}, 032435 (2026).

\bibitem{Mortensen_2018}
N. A. Mortensen, P. A. D. Gonçalves, M. Khajavikhan, D. N. Christodoulides, C. Tserkezis, and C. Wolff, {\it Fluctuations and noise-limited sensing near the exceptional point of parity-time-symmetric resonator systems}, Optica \textbf{5}, 1342 (2018). 

\bibitem{Aspelmeyer_2014}
M. Aspelmeyer, T. J. Kippenberg, and F. Marquardt, {\it Cavity optomechanics}, Rev. Mod. Phys. \textbf{86}, 1391 (2014).

\bibitem{Jing_2017}
H. Jing, \c{S}. K. \"Ozdemir, H. L\"u, and F. Nori, {\it High-order exceptional points in optomechanics}, Sci. Rep. \textbf{7}, 3386 (2017).

\bibitem{Ghosh_2026}
A. Ghosh and M. Bhattacharya, {\it Quantum jumps in open cavity optomechanics and Liouvillian versus Hamiltonian exceptional points}, Phys. Rev. A \textbf{113}, 053530 (2026).

\bibitem{Xu_2016}
H. Xu, D. Mason, L. Jiang, and J. G. E. Harris, {\it Topological energy transfer in an optomechanical system with exceptional points}, Nature \textbf{537}, 80 (2016).

\bibitem{Shiralieva_2026}
A. Shiralieva, G. A. Starkov, and B. Trauzettel,
{\it Multiblock exceptional points in open quantum systems}, Phys. Rev. A \textbf{113}, 052202 (2026).

\bibitem{Wang_2024}
C. Wang, L. Banniard, K. B{\o}rkje, F. Massel, L. M. de L\'epinay, and M. A. Sillanp\"a\"a, {\it Ground-state cooling of a mechanical oscillator by a noisy environment}, Nat. Commun. \textbf{15}, 7395 (2024). 

\bibitem{Kalman_1960}
R. E. Kalman, {\it On the general theory of control systems}, In: Proceedings of the First IFAC Congress, Moscow (Butterworths, London, 1961).

\bibitem{HoKalman_1966}
B. L. Ho and R. E. Kalman, {\it Effective construction of linear state-variable models from input/output functions},
Regelungstechnik \textbf{14}, 545 (1966).

\bibitem{GroblacherCooling_2009}
S. Gr\"oblacher, J. B. Hertzberg, M. R. Vanner, G. D. Cole, S. Gigan, K. C. Schwab, and M. Aspelmeyer, {\it Demonstration of an ultracold micro-optomechanical oscillator in a cryogenic cavity}, Nat. Phys. \textbf{5}, 485 (2009).

\bibitem{GroblacherStrongCoupling_2009}
S. Gr\"oblacher, K. Hammerer, M. R. Vanner, and M. Aspelmeyer, {\it Observation of strong coupling between a micromechanical resonator and an optical cavity field}, Nature \textbf{460}, 724 (2009).

\bibitem{Weaver_2017}
M. J. Weaver, F. Buters, F. Luna, H. Eerkens, K. Heeck, S. de Man, and D. Bouwmeester, {\it Coherent optomechanical state transfer between disparate mechanical resonators}, Nat. Commun. \textbf{8}, 824 (2017). 

\end{thebibliography}

\begin{thebibliography}{99}

\bibitem{KatoSupp}
T. Kato, {\it Perturbation theory for linear operators}, 2nd ed. (Springer, Berlin, 1995).

\bibitem{HornJohnsonSupp}
R. A. Horn and C. R. Johnson, {\it Matrix analysis}, 2nd ed.
(Cambridge University Press, Cambridge, 2013).

\bibitem{GoodmanSupp}
N. R. Goodman, {\it Statistical analysis based on a certain multivariate complex Gaussian distribution (An introduction)}, Ann. Math. Statist. \textbf{34}, 152 (1963).

\bibitem{SeberLeeSupp}
G. A. F. Seber and A. J. Lee,
{\it Linear regression analysis}, 2nd ed. (Wiley, Hoboken, 2003).

\bibitem{KesslerSupp}
H. Ke{\ss}ler, J. Klinder, M. Wolke, and A. Hemmerich, {\it Optomechanical atom-cavity interaction in the sub-recoil regime}, New J. Phys. \textbf{16}, 053008 (2014).

\bibitem{KlinderSupp}
J. Klinder, H. Ke{\ss}ler, M. Wolke, L. Mathey, and A. Hemmerich, {\it Dynamical phase transition in the open Dicke model}, Proc. Natl. Acad. Sci. U.S.A. \textbf{112}, 3290 (2015).

\bibitem{WolkeSupp}
M. Wolke, J. Klinner, H. Ke{\ss}ler, and A. Hemmerich, {\it Cavity cooling below the recoil limit}, Science \textbf{337}, 75 (2012).

\end{thebibliography}
\end{document}